\documentclass{llncs}
\usepackage{version}

\excludeversion{conf}
\includeversion{report}

\usepackage{graphicx}
\usepackage{amsmath, amsfonts,amssymb}
\usepackage{enumerate}
\usepackage{mdframed}
\usepackage[ruled, vlined]{algorithm2e}
\usepackage[dvipsnames]{xcolor}

\newcommand{\siglists}{{\Sigma_{\it list}}}
\newcommand{\sigpairs}{{\Sigma_{\it pair}}}
\newcommand{\siglistspairs}{{\Sigma_{\it lp}}}
\newcommand{\structsort}{{\bf struct}}
\newcommand{\listsort}{{\bf list}}
\newcommand{\elemsort}{{\bf elem}}
\newcommand{\pairsort}{{\bf pair}}
\newcommand{\elemsorts}{{\bf {Elem}}}
\newcommand{\structsorts}{{\bf {Struct}}}
\newcommand{\witness}{{\it wtn}}
\newcommand{\iwitness}{{\it wtn_{i}}}
\newcommand{\fwitness}{{\it wtn_{f}}}
\newcommand{\depth}{{\it depth}}

\newcommand{\smtlib}{SMT-LIB~2\xspace}

\newcommand{\consequence}[1]{\vdash_{{#1}}}
\newcommand{\valid}[1]{\vdash_{{#1}}}

\newcommand{\Is}{{\it Is}}

\newcommand{\Ind}{{\it Ind}}
\newcommand{\Fin}{{\it Fin}}

\newcommand{\constructors}{{\Ca}{\Oa}}
\newcommand{\consig}[1]{{#1}_{\mid \constructors}}
\newcommand{\selectors}{{\Sa}{\Ea}}
\newcommand{\predecessors}{{\Pa}}

\newcommand{\TREE}{{\it TREE}}
\newcommand{\canon}[1]{[[{#1}]]}
\newcommand{\qf}{{\it QF}}

\newcommand{\powerset}[1]{{{\cal P}({#1})}}

\newcommand{\fv}[2]{{\it vars}_{#1}({#2})}
\newcommand{\sorts}[1]{{\cal S}_{#1}}
\newcommand{\func}[1]{{\cal F}_{#1}}

\newcommand{\pred}[1]{{\cal P}_{#1}}

\newcommand{\Ta}{{\cal T}}

\newcommand{\Aa}{{\cal A}}
\newcommand{\Ba}{{\cal B}}
\newcommand{\Ca}{{\cal C}}
\newcommand{\Sa}{{\cal S}}

\newcommand{\Oa}{{\cal O}}
\newcommand{\Ea}{{\cal E}}
\newcommand{\Pa}{{\cal P}}
\newcommand{\set}[1]{\left\{#1\right\}}

\newcommand{\card}[1]{{\left |#1\right |}}

\newcommand{\A}{{\cal A}}

\newcommand{\w}{\wedge}

\newcommand{\ra}{\rightarrow}
\newcommand{\bi}{\begin{itemize}}
\newcommand{\ei}{\end{itemize}}
\newcommand{\be}{\begin{enumerate}}
\newcommand{\ee}{\end{enumerate}}
\newcommand{\bd}{\begin{description}}
\newcommand{\ed}{\end{description}}

\newcommand{\ora}{\overrightarrow}

\newcommand{\fe}{\varphi}

\newcommand{\suq}{\subseteq}
\newcommand{\st}{{\ |\ }}
\newcommand{\til}{,\dots,}

\usepackage[inline]{enumitem}
\usepackage{hyperref}
\usepackage{cleveref}
\renewenvironment{proof}{\noindent{\it Proof}\/:}{\qed\hspace*{1pt}}

\begin{document}
\title{Politeness for the Theory of\\ Algebraic Datatypes
  \thanks{This project was partially supported by a grant from the Defense
Advanced Research Projects Agency
(N66001-18-C-4012), the Stanford CURIS program, and
Jasmin Blanchette's European Research Council (ERC)
    starting grant Matryoshka (713999).
  }
}
%
%
\author{
Ying Sheng\inst{1}\orcidID{0000-0002-1883-2126}
\and
Yoni Zohar\inst{1}\orcidID{0000-0002-2972-6695} 
\and
Christophe Ringeissen\inst{2}\orcidID{0000-0002-5937-6059}
\and
Jane Lange\inst{1}\orcidID{0000-0002-0642-9815}
\and\\
Pascal Fontaine\inst{2,3}\orcidID{0000-0003-4700-6031}
\and
Clark Barrett\inst{1}\orcidID{0000-0002-9522-3084}
}
\authorrunning{Sheng et\ al.}
%
\institute{Stanford University\and
Universit\'e de Lorraine, CNRS, Inria, LORIA, F-54000 Nancy, France\and
Universit\'e de Li\`ege, Belgium
}
\maketitle              
\setcounter{footnote}{0}
\begin{abstract}
Algebraic datatypes, and among them lists and trees, have attracted a lot of interest in
automated reasoning and Satisfiability Modulo Theories (SMT).  Since its latest
stable version, the SMT-LIB standard defines a theory of algebraic datatypes,
which is currently supported by several mainstream SMT solvers.  
In this paper, we study this particular theory of datatypes and prove that it is strongly polite, 
showing also how it can be combined with other arbitrary disjoint theories using polite combination.  
Our results cover both inductive and finite datatypes, as well as their union.
The combination method uses a new, simple, and natural notion of additivity,
that enables deducing strong politeness from (weak) politeness.
\end{abstract}

\section{Introduction}
\label{sec:intro}
Algebraic datatypes such as lists and trees are extremely common
in many programming languages.
Reasoning about them is therefore crucial for modeling and verifying
programs.
For this reason, various decision procedures for algebraic datatypes
have been, and continue to be 
developed and employed by formal reasoning tools such as
theorem provers and Satisfiability Modulo Theories (SMT) solvers.
For example, the general algorithm of \cite{BST07} describes a decision
procedure for datatypes suitable for SMT solvers.
Consistently with the SMT paradigm, \cite{BST07} leaves the combination
of datatypes with other theories to general combination methods,
and focuses on parametric datatypes (or {\em generic} datatypes as they are
called in the programming languages community).

The traditional combination method of Nelson and Oppen 
\cite{DBLP:journals/toplas/NelsonO79}
is applicable for the combination of this theory with many other theories,
as long as  the other theory is {\em stably infinite} (a technical
condition that intuitively amounts to the ability to extend every
model to an infinite one).
Some theories of interest, however, are not stably infinite, 
the most notable one being the theory of fixed-width bit-vectors, which
is commonly used for modeling and verifying both hardware and software.
To be able to perform combinations with such theories, 
a more general combination method was designed \cite{10.1007/11559306_3},
which relies on {\em polite theories}.
Roughly speaking, a theory is polite if: $(i)$ every model 
can be arbitrarily enlarged; 
and $(ii)$ there is a {\em witness}, a function that 
transforms any quantifier-free formula to an equivalent
quantifier-free formula such that if
the original formula is satisfiable,
the new formula is satisfiable in a ``minimal'' interpretation.
This notion was later strengthened to
{\em strongly polite theories}~\cite{JBLPAR},
which also account for possible arrangements of the variables
in the formula.
Strongly polite theories can be combined with any other disjoint decidable theory, even if
that other theory is not stably infinite.
While strong politeness was already proven for several useful theories
(such as equality, arrays, sets, multisets \cite{10.1007/11559306_3}),
strong politeness of algebraic datatypes remained an unanswered question.

The main contribution of this paper is an
affirmative answer to this question. 
We introduce a {\em witness} function
that essentially 
``guesses" the right constructors of variables without an explicit
constructor in the formula.
We show how to ``shrink" any model of a formula that
is the output of this function into a minimal model.
The witness function, as well as the model-construction,
can be used by any SMT solver for the theory of datatypes
that implements polite theory combination.
We introduce and use the notion of additive witnesses, 
which allows us to prove politeness and conclude
strong politeness.
We further study the theory of datatypes beyond politeness and extend a
decision procedure for a subset of this theory 
presented in \cite{CFR19}
to support the full theory.

\subsection*{Related Work}
The theory investigated in this paper is that of algebraic datatypes,
as defined by the \smtlib standard \cite{SMTLib2017}.
Detailed information on this theory, including a decision procedure
and related work,
can be found in \cite{BST07}.  Later work extends this procedure
to handle
shared selectors \cite{DBLP:conf/cade/ReynoldsVBTB18} and co-datatypes \cite{DBLP:journals/jar/ReynoldsB17}.
More recent approaches for solving formulas about datatypes
use, e.g., theorem provers \cite{DBLP:conf/popl/KovacsRV17},
variant satisfiability \cite{DBLP:conf/lopstr/GutierrezM17,DBLP:journals/scp/Meseguer18}, and reduction-based decision procedures \cite{8531279,DBLP:journals/entcs/BonacinaE07,DBLP:journals/tocl/ArmandoBRS09}. 

In this paper, we focus on polite theory combination.
Other combination methods for non stably infinite theories include
shiny theories \cite{DBLP:journals/jar/TinelliZ05}, 
gentle theories \cite{DBLP:conf/frocos/Fontaine09}, and
parametric theories \cite{DBLP:conf/tacas/KrsticGGT07}.
The politeness property was introduced in 
\cite{10.1007/11559306_3}, and 
extends the stable infiniteness assumption initially used by
Nelson and Oppen. 
Polite theories can
be combined \`a la Nelson-Oppen with any arbitrary decidable theory. 
Later, a flaw in the original definition of politeness was found~\cite{JBLPAR},
and a corrected definition (here called {\em strong politeness}) was introduced.
%
Strongly polite theories were further studied in \cite{Casal2018}, where the
authors proved their
equivalence with shiny theories \cite{DBLP:journals/jar/TinelliZ05}.

More recently, it was proved~\cite{CFR19} that a general family
of datatype theories extended with bridging functions is strongly polite. 
This includes the theories of lists/trees with length/size functions.
The authors also proved that a class of axiomatizations of datatypes is 
strongly polite. 
In contrast,
in this paper we focus on standard interpretations, as defined by the \smtlib
standard, without any size function, but including selectors and testers.
One can notice that the theory of
standard lists without the length function, and more generally the theory of finite trees
without the size function, were not mentioned 
as polite in a recent
survey \cite{DBLP:conf/birthday/BonacinaFRT19}.
Actually, it was unclear to the authors of
\cite{DBLP:conf/birthday/BonacinaFRT19} whether these theories are strongly polite. This
is now clarified in the current paper.

\subsection*{Outline}
The paper is organized as follows.
\Cref{sec:prelim} provides the necessary notions from
first-order logic and polite theories.
\Cref{sec:smtlib} provides our working definition
of the theory of datatypes, which is based on \smtlib.
\Cref{secweakstrong} discusses the difference between politeness 
and strong politeness, and introduces
a useful condition for their equivalence.
\Cref{sec:politedt} contains the main result of this paper,
namely that the theory of algebraic datatypes is strongly polite.
\Cref{sec:axiomatizations} studies various axiomatizations
of the theory of datatypes, and relates them
to politeness.
\Cref{sec:conc} concludes with directions for further research.

\section{Preliminaries}
\label{sec:prelim}
\subsection{Signatures and Structures}
We briefly review usual definitions of
many-sorted first-order logic with equality (see
\cite{enderton2001mathematical,10.1007/978-3-540-30227-8_53}
for more details).
For any set $S$, 
an {\em $S$-sorted set} $A$ is a function from $S$ to $\powerset{X}\setminus\set{\emptyset}$ for some
set $X$ (i.e., $A$ assigns a non-empty set to every element of $S$),
such that $A(s)\cap A(s')=\emptyset$ whenever
$s\neq s'$.
We use $A_{s}$ to denote $A(s)$ for every $s\in S$, and call the elements
of $S$ {\em sorts}.
When there is no ambiguity, we sometimes treat sorted sets as sets 
(e.g., when writing expressions like $x\in A$).
%
Given a set $S$ (of sorts), the {\em canonical $S$-sorted set},
denoted $\canon{S}$,
satisfies $\canon{S}_{s}=\set{s}$ for every $s\in S$.
A {\em many-sorted signature} $\Sigma$ consists of
a set $\sorts{\Sigma}$ (of {\em sorts}),
a set $\func{\Sigma}$ of function symbols,
and a set $\pred{\Sigma}$ of predicate symbols.
Function symbols have arities of the form $\sigma_{1}\times\ldots\times\sigma_{n}\ra\sigma$,
and predicate symbols have arities of the form $\sigma_{1}\times\ldots\times\sigma_{n}$,
with $\sigma_{1}\til\sigma_{n},\sigma\in\sorts{\Sigma}$.
For each sort $\sigma\in\sorts{\Sigma}$,
$\pred{\Sigma}$ includes an {\em equality symbol} $=_{\sigma}$ of arity
$\sigma\times\sigma$.
We denote it by $=$ when $\sigma$ is clear from context.
$\Sigma$ is called {\em finite} if $\sorts{\Sigma}$, $\func{\Sigma}$,
and $\pred{\Sigma}$ are finite.

We assume an underlying $\sorts{\Sigma}$-sorted set 
of {\em variables}.
Terms, formulas, and literals are defined in the usual way.
For a $\Sigma$-formula $\phi$ and a sort $\sigma$, we denote the set of 
free variables in $\phi$ of sort $\sigma$ by $\fv{\sigma}{\phi}$.
This notation naturally extends to $\fv{S}{\phi}$ when $S$ is a set of sorts.
A sentence is a formula without free variables.
We denote by $\qf(\Sigma)$ the set of quantifier-free formulas of $\Sigma$.
A $\Sigma$-literal is called {\em flat} if it has one of the following forms: 
$x=y$, $x\neq y$, $x=f(x_{1}\til x_{n})$, 
$P(x_{1}\til x_{n})$, or $\neg P(x_{1}\til x_{n})$ 
for some variables $x,y,x_{1}\til x_{n}$ and  function and predicate symbols
$f$ and $P$ from $\Sigma$.

A {\em $\Sigma$-structure} is a many-sorted structure for $\Sigma$, without interpretation of variables.
It consists of a $\sorts{\Sigma}$-sorted set $A$, and interpretations
to the function and predicate symbols of $\Sigma$.  
We further require that  $=_{\sigma}$ is interpreted
as the identity relation over $A_\sigma$ for every $\sigma\in\sorts{\Sigma}$.
A {\em $\Sigma$-interpretation} $\Aa$ is an extension of
a $\Sigma$-structure with interpretations to some set of variables.  
For any $\Sigma$-term $\alpha$, $\alpha^{\Aa}$
denotes the interpretation of $\alpha$ in $\Aa$.
When $\alpha$ is a set of $\Sigma$-terms, $\alpha^{\Aa}=\set{x^{\Aa}\mid x\in \alpha}$.
Similarly, $\sigma^{\Aa}$, $f^{\Aa}$ and $P^{\Aa}$ denote the interpretation
of $\sigma$, $f$ and $P$ in $\Aa$.
Satisfaction is defined as usual.  
$\Aa \models \varphi$ denotes that $\Aa$ satisfies $\varphi$.

A $\Sigma$-theory $T$ is a class of $\Sigma$-structures.
A $\Sigma$-interpretation 
whose variable-free part is in $T$ is called
a $T$-interpretation.
A $\Sigma$-formula $\phi$ is $T$-satisfiable if $\Aa\models\phi$ for some $T$-interpretation $\Aa$.
Two formulas $\phi$ and $\psi$ are {\em $T$-equivalent} if they are satisfied by the same class of $T$-interpretations.
Let $\Sigma_{1}$ and $\Sigma_{2}$ be signatures,
$T_{1}$ a $\Sigma_{1}$-theory, and $T_{2}$ a $\Sigma_{2}$-theory.
The combination of $T_{1}$ and $T_{2}$, denoted 
$T_{1}\oplus T_{2}$,
 is the class of $\Sigma_{1}\cup\Sigma_{2}$-structures $\Aa$
 such that $\Aa^{\Sigma_{1}}$ is in $T_{1}$ and $\Aa^{\Sigma_{2}}$ is in $T_{2}$,
where $\Aa^{\Sigma_{i}}$ is the restriction of $\Aa$ to $\Sigma_{i}$ for
$i\in\set{1,2}$.

\subsection{The \smtlib Theory of Datatypes}
\label{sec:smtlib}

In this section we formally define the \smtlib theory of algebraic
datatypes.
The formalization is based on \cite{SMTLib2017}, but
is adjusted to suit our investigation of 
politeness.

\begin{definition}
\label{def:trees}
Given a signature $\Sigma$,
a set 
$S\suq\sorts{\Sigma}$ and
an $S$-sorted set $A$, 
the set of {\em $\Sigma$-trees} over $A$ of sort $\sigma \in\sorts{\Sigma}$ is denoted by 
$T_{\sigma}(\Sigma,A)$ and is inductively defined as follows:
  \begin{itemize}
  \item $T_{\sigma,0}(\Sigma,A) = A_\sigma$ if $\sigma\in S$ and $\emptyset$ otherwise.
  \item $T_{\sigma,i+1}(\Sigma,A) = T_{\sigma,i}(\Sigma,A) \cup \{ c(t_1,\dots,t_n) ~|~ c:\sigma_{1}\times\ldots\times\sigma_{n} \ra\sigma \in \func{\Sigma},  t_j \in T_{\sigma_j,i}(\Sigma,A) \mbox{ for } j=1,\dots,n \}$ for each $i \geq 0$.
  \end{itemize}
Then $T_{\sigma}(\Sigma,A) = \bigcup_{i \geq 0} T_{\sigma,i}(\Sigma,A)$.
The {\em depth} of a $\Sigma$-tree over $A$ is inductively defined by
$\depth(a)=0$ for every $a\in A$,
$\depth(c)=1$ for every $0$-ary function symbol $c\in\func{\Sigma}$,
and $\depth(c(t_{1}\til t_{n}))=1+max(\depth(t_{1})\til \depth(t_{n}))$
for every $n$-ary function symbol $c$ of $\Sigma$.
\end{definition}  

The idea behind \Cref{def:trees} is that
$T_{\sigma}(\Sigma,A)$ contains all ground $\sigma$-sorted terms
constructed from the elements of $A$ (considered as constant 
symbols) and the function symbols of $\Sigma$.

\begin{example}
\label{exm:tree}
Let $\Sigma$ be a signature with two sorts, $\elemsort$ and $\structsort$, and whose function symbols are $b$ of arity $\structsort$, and 
$c$ of arity $(\elemsort \times\structsort \times\structsort )\ra\structsort$.
Consider the $\{ \elemsort \}$-sorted set $A = \{ a \}$. For the $\elemsort$ sort, $T_{\elemsort}(\Sigma,A)$ is the singleton $A = \{ a \}$ and the $\Sigma$-tree $a$ is of depth $0$. For the $\structsort$ sort, 
$T_{\structsort}(\Sigma,A)$ includes infinitely many $\Sigma$-trees, such as $b$ of depth $1$, $c(a,b,b)$ of depth $2$, and $c(a,c(a,b,b),b)$ of depth $3$.
\end{example}

\begin{definition}
\label{def:IDT_sig}
A finite signature $\Sigma$ is called a {\em datatypes signature}
if
$\sorts{\Sigma}$ is the disjoint union of two sets of sorts
$\sorts{\Sigma}=\elemsorts_{\Sigma}\uplus\structsorts_{\Sigma}$ and
$\func{\Sigma}$ is the disjoint union of two sets of function symbols 
$\func{\Sigma}=\constructors_{\Sigma}\uplus\selectors_{\Sigma}$,
such that
$\selectors_{\Sigma}=
\{s_{c,i}:\sigma\ra\sigma_{i} \mid
c\in\constructors_{\Sigma},
c:\sigma_{1}\til\sigma_{n}\ra\sigma,
1\leq i\leq n
\}$ and
$\predecessors_{\Sigma}=
\{is_{c}:\sigma \mid
c\in\constructors_{\Sigma},
c:\sigma_{1}\til\sigma_{n}\ra\sigma
\}$.
We denote
by
$\consig{\Sigma}$ the signature
with the same sorts as $\Sigma$, no predicate symbols (except $=_{\sigma}$ for $\sigma\in\sorts{\Sigma}$),
and whose function symbols are $\constructors_{\Sigma}$.
We further require the following well-foundedness
requirement:
$T_{\sigma}(\consig{\Sigma},\canon{\elemsorts_{\Sigma}}) \neq \emptyset$ for any 
$\sigma \in \structsorts_{\Sigma}$.
\end{definition}

From now on, we omit the subscript $\Sigma$ from the above notations
(e.g., when writing $\canon{\elemsorts}$ rather than $\canon{\elemsorts_{\Sigma}}$, $\constructors$ rather than $\constructors_{\Sigma}$) 
whenever
$\Sigma$ is clear from the context.
Notice that \Cref{def:IDT_sig} remains equivalent
if we replace $\canon{\elemsorts}$ by any 
(non-empty) $\elemsorts$-sorted set $A$.
The set $\canon{\elemsorts}$ has been chosen since this minimal $\elemsorts$-sorted set is sufficient.

In accordance with \smtlib, 
we call the elements of $\constructors$ {\em constructors},
the elements of $\selectors$ {\em selectors},
and the elements of $\predecessors$ {\em testers}.
$0$-ary constructors are called {\em nullary}.
In what follows, $\Sigma$ denotes an arbitrary datatypes signature. 

In the next example we review some common datatypes signatures.

\begin{example}
\label{exm:dt}
The signature $\siglists$ has two sorts, $\elemsort$ and $\listsort$.
Its function symbols are $cons$ of arity $(\elemsort\times\listsort)\ra\listsort$,
$nil$ of arity $\listsort$, $car$ of arity $\listsort\ra\elemsort$
and $cdr$ of arity $\listsort\ra\listsort$.
Its predicate symbols 
are $is_{nil}$ and $is_{cons}$, both of arity $\listsort$.
It is a datatypes signature, with $\elemsorts=\set{\elemsort}$,
$\structsorts=\set{\listsort}$,
$\constructors=\set{nil,cons}$ and
$\selectors=\set{car,cdr}$.
It is often used to model lisp-style linked lists.
$car$ represents the head of the list and
$cdr$ represents its tail.
$nil$ represents the empty list.
$\siglists$ is well-founded as
$T_{\listsort}(\consig{\siglists},\canon{\elemsorts})$ 
includes $nil$.

The signature $\sigpairs$ also has two sorts, $\elemsort$ and $\pairsort$.
Its function symbols are $pair$ of arity $(\elemsort\times\elemsort)\ra\pairsort$
and
$first$ and $second$ of arity $\pairsort\ra\elemsort$.
Its predicate symbol is $is_{pair}$ of arity $\pairsort$.
It is a datatypes signature, with $\elemsorts=\set{\elemsort}$,
$\structsorts=\set{\pairsort}$,
$\constructors=\set{pair}$,
and
$\selectors=\set{first,second}$.
It can be used to model ordered pairs, together with 
projection functions.
It is well-founded as
$T_{\pairsort}(\consig{\sigpairs},\canon{\elemsorts})$ 
is not empty (as $\canon{\elemsorts}$ is not empty).

The signature $\siglistspairs$ has three sorts, $\elemsort$, $\pairsort$
and $\listsort$.
Its function symbols are 
$cons$ of arity $(\pairsort\times\listsort)\ra\listsort$,
$car$ of arity $\listsort\ra\pairsort$,
as well as $nil, cdr, first, second$ with arities as above.
Its predicate symbols are $is_{pair}$, $is_{cons}$ and $nil$, with
arities as above.
It can be used to model lists of ordered pairs.
Similarly to the above signatures, it is a datatypes signature.
\end{example}

Next, we distinguish between finite datatypes (e.g., records) and
inductive datatypes (e.g., lists).

\begin{definition}
\label{def:ind_fin_sort}
A sort $\sigma \in \structsorts$ is {\em finite} if 
$T_\sigma(\consig{\Sigma},\canon{\elemsorts})$ is finite,
and is called {\em inductive} otherwise.
\end{definition}

We denote the set of inductive sorts in $\Sigma$ by $\Ind(\Sigma)$
and the set of its finite sorts by $\Fin(\Sigma)$.
Note that if $\sigma$ is inductive, then according to \Cref{def:trees,def:ind_fin_sort}
we have that 
for any natural number $i$ there exists a natural number $i' > i$ such that 
$T_{\sigma,i'}(\consig{\Sigma},\canon{\elemsorts}) \neq T_{\sigma,i}(\consig{\Sigma},\canon{\elemsorts})$. 
Further, 
for any natural number $d$ and every $\elemsorts$-sorted set $D$
there exists a natural number $i'$ such that
$T_{\sigma,i'}(\consig{\Sigma},D)$ contains an element
whose depth is greater than $d$.

\begin{example}
$\listsort$ is inductive 
in $\siglists$ and $\siglistspairs$.
$\pairsort$ is finite 
in $\sigpairs$ and $\siglistspairs$.
\end{example}

Finally, we define datatypes structures and the theory
of algebraic datatypes.

\begin{definition}
  Let $\Sigma$ be a datatypes signature and $D$ an $\elemsorts$-sorted set.
A $\Sigma$-structure $\Aa$ is said to be a {\em datatypes $\Sigma$-structure generated by $D$} if:
\bi
\item $\sigma^{\Aa} = T_{\sigma}(\consig{\Sigma},D)$ for every sort $\sigma \in \sorts{\Sigma}$,
\item $c^{\Aa}(t_1 \til t_n) = c(t_1 \til t_n)$ for every $c\in\constructors$ of 
arity $(\sigma_{1}\times\ldots\times\sigma_{n})\ra\sigma$ and $t_{1} \in \sigma_{1}^{\Aa} \til t_{n} \in \sigma_{n}^{\Aa}$,
\item $s_{c,i}^{\Aa}( c(t_{1} \til t_{n}) )=t_{i}$ 
for every $c\in\constructors$ of 
arity $(\sigma_{1}\times\ldots\times\sigma_{n})\ra\sigma$, $t_{1} \in \sigma_{1}^{\Aa} \til t_{n} \in\sigma_{n}^{\Aa}$ and $1\leq i\leq n$,
\item $is_{c}^{\Aa}=\set{ c(t_{1} \til t_{n}) \mid t_{1} \in \sigma_{1}^{\Aa} \til t_{n} \in \sigma_{n}^{\Aa} }$ for every $c\in\constructors$ of arity
$(\sigma_{1}\times\ldots\times\sigma_{n})\ra\sigma$.
\ei
$\Aa$ is said to be a {\em datatypes $\Sigma$-structure} if it is a datatypes $\Sigma$-structure
generated by $D$ for some $\elemsorts$-sorted set $D$. 
The {\em $\Sigma$-theory of datatypes}, denoted
$\Ta_{\Sigma}$ is the class of datatypes $\Sigma$-structures.
\end{definition}

Notice that the interpretation of selector functions $s_{c,i}$ 
when applied to terms that are constructed using a constructor
different than $c$ is not fixed and can be set arbitrarily
in datatypes structures, consistently with \smtlib.

\begin{example}
If $\Aa$ is a datatypes $\siglists$-structure then
$\listsort^{\Aa}$ is the set of terms constructed from
$\elemsort^{\Aa}$ and $cons$, plus $nil$.
If $\elemsort^{\Aa}$ is the set of natural numbers, 
then $\listsort^{\Aa}$ contains, e.g.,
$nil$, $cons(1, nil)$, and $cons(1, cons (1, cons (2, nil)))$.
These correspond to the lists $[]$ (the empty list),
$[1]$ and $[1,1,2]$, respectively.

If $\Aa$ is a datatypes $\sigpairs$-structure then
$\pairsort^{\Aa}$ is the set of terms of the form
$pair(a,b)$ with $a,b\in\elemsort^{\Aa}$.
If $\elemsort^{\Aa}$ is again interpreted as the set of natural numbers,
$\pairsort^{\Aa}$ includes, for example, the terms
$pair(1,1)$ and $pair(1,2)$, that correspond to
$(1,1)$ and $(1,2)$, respectively.
Notice that in this case, $\pairsort^{\Aa}$ is an infinite set
even though $\pairsort$ is a finite sort (in terms of \Cref{def:ind_fin_sort}).

Datatypes $\siglistspairs$-structures with the same
interpretation for $\elemsort$ 
include the terms
$nil$, $cons(pair(1,1), nil)$, and 
$cons(pair(1,1), cons(pair(1,2),nil))$
in the interpretation for $\listsort$,
that correspond to $[]$, $[(1,1)]$ and $[(1,1),(1,2)]$, respectively.
If we rename $\elemsort$ in the definition of $\siglists$
to $\pairsort$,
we get that
$\Ta_{\siglistspairs}=\Ta_{\siglists}\oplus\Ta_{\sigpairs}$.
\end{example}

\subsection{Polite Theories}
\label{prelim:polite}

Given two theories $T_1$ and $T_2$, a combination method \`a la
Nelson-Oppen provides a modular way to decide $T_1 \cup T_2$-satisfiability
problems using the satisfiability procedures known for $T_1$ and
$T_2$. Assuming that $T_1$ and $T_2$ have disjoint signatures is not
sufficient to get a complete combination method for the satisfiability
problem. The reason is that $T_1$ and $T_2$ may share sorts, and the equality
symbol on these shared sorts. To be complete, $T_1$ and $T_2$ must agree on the cardinality
of their respective models, and there must be an agreement between
$T_1$ and $T_2$ on the interpretation of shared formulas built over
the equality symbol. These two requirements can be easily fulfilled, based on the following definitions: 

\begin{definition}[Stable Infiniteness]
Given a signature $\Sigma$ and
a set 
$S\suq\sorts{\Sigma}$, we say that a $\Sigma$-theory
$T$ is {\em stably infinite with respect to $S$}
if every quantifier-free $\Sigma$-formula 
that is $T$-satisfiable is also $T$-satisfiable
by a $T$-interpretation $\Aa$ in which $\sigma^{\Aa}$
is infinite for every $\sigma\in S$.
\end{definition}

\begin{definition}[Arrangement]
Let 
$V$ be a finite set of variables whose sorts are in $S$ and  
$\set{V_{\sigma}\st \sigma\in S}$ a partition of $V$ such that $V_{\sigma}$ is the set of variables of sort $\sigma$ in $V$.
We say that a formula $\delta$ is an {\em arrangement of $V$}
if 
$\delta=\bigwedge_{\sigma\in S}(\bigwedge_{(x,y)\in E_{\sigma}}(x=y)\w\bigwedge_{(x,y)\notin E_{\sigma}}(x\neq y))$, where $E_{\sigma}$ is some equivalence relation over $V_{\sigma}$ for each $\sigma\in S$.

\end{definition}

Assume that both $T_1$ and $T_2$ are stably infinite with disjoint
signatures, and let $V$ be the finite set of variables shared by $T_1$
and $T_2$.  Under this assumption, $T_1$ and $T_2$ can agree on an
infinite cardinality, and guessing an arrangement of $V$ suffices to
get an agreement on the interpretation of shared formulas.

In this paper we are interested in an asymmetric disjoint combination
where $T_1$ and $T_2$ are not both stably infinite. In this scenario,
one theory can be arbitrary. As a counterpart, the other theory must
be more than stably infinite: it must be polite, meaning that it is
always possible to increase the cardinality of a model and to have a
model whose cardinality is finite.

In the following we
decompose the politeness definition from \cite{10.1007/11559306_3,JBLPAR}
in order to distinguish between politeness and strong 
politeness (in terms of \cite{Casal2018}) in 
various levels of the definition.
In what follows, $\Sigma$ is an arbitrary (many-sorted) signature, 
$S\suq\sorts{\Sigma}$,
and $T$ is a $\Sigma$-theory.

\begin{definition}[Smooth]
The theory $T$ is \,{\em smooth} w.r.t.\ $S$ if for every quantifier-free formula $\phi$,  $T$-interpretation $\Aa$ that satisfies $\phi$, and function 
$\kappa$ from $S$ to the class of cardinals such that $\kappa(\sigma)\geq\card{\sigma^{\Aa}}$ for every $\sigma\in S$ 
there exists a $\Sigma$-interpretation $\Aa'$ that satisfies $\phi$ with $\card{\sigma^{\Aa'}}=\kappa(\sigma)$ for every $\sigma\in S$.
\end{definition}

In definitions introduced above, as well as below, we often identify
singletons with their single elements when there is no ambiguity 
(e.g., when saying that a theory is smooth w.r.t.\ a sort $\sigma$).

We now introduce some concepts in order to define finite witnessability.
Let $\phi$ be a quantifier-free $\Sigma$-formula and
$\Aa$ a $\Sigma$-interpretation.
We say that $\Aa$ {\em finitely witnesses $\phi$ for $T$ w.r.t.\ $S$} (or, is a {\em finite witness of $\phi$ for $T$ w.r.t.\ $S$}), if $\Aa$ is a $T$-interpretation,
$\Aa\models\phi$, and 
$\sigma^{\Aa}=\fv{\sigma}{\phi}^{\Aa}$
 for every $\sigma\in S$.
We say that $\phi$ is 
{\em  finitely witnessed for $T$ w.r.t.\ $S$} if
it is either $T$-unsatisfiable or it has a finite witness for $T$ w.r.t.\ $S$.
$\phi$ is 
{\em strongly finitely witnessed for $T$ w.r.t.\ $S$} if
$\phi\w\delta_{V}$ is  finitely witnessed for $T$ w.r.t.\ $S$ for every arrangement $\delta_{V}$ of $V$, where $V$ is any set of variables whose sorts are in $S$.
We say that a function $\witness : \qf(\Sigma) \ra \qf(\Sigma)$ is
a {\em (strong) witness for $T$ w.r.t.\ $S$} if for every $\phi\in \qf(\Sigma)$ we have that:
\begin{enumerate*}
\item $\phi$ and $\exists\,\ora{w}.\:\witness(\phi)$ are $T$-equivalent for $\ora{w}=\fv{}{\witness(\phi)}\setminus\fv{}{\phi}$; and
\item $\witness(\phi)$ is (strongly) finitely witnessed for $T$ w.r.t.\ $S$.\footnote{
We note that in practice,
the new variables in $\witness(\phi)$ are assumed to
be fresh not only with respect to $\phi$, but also with respect to the
formula from the second theory being combined.
}
\end{enumerate*}

\begin{definition}[Finitely Witnessable]
The theory $T$ is {\em (strongly) finitely witnessable} w.r.t.\ $S$ if there exists a (strong) witness for $T$ w.r.t.\ $S$ which is computable.
\end{definition}

\begin{definition}[Polite]
$T$ is called {\em (strongly) polite w.r.t.\ $S$} if it is smooth and (strongly) finitely witnessable w.r.t.\ $S$.
\end{definition}

Finally, we
recall the following theorem from \cite{JBLPAR}.

\begin{theorem}[\cite{JBLPAR}]
\label{thm:lpar}
Let $\Sigma_1$ and $\Sigma_2$ be signatures
and let $S=\sorts{\Sigma_1}\cap \sorts{\Sigma_2}$. If
   $T_1$ is a $\Sigma_1$-theory strongly polite w.r.t.\ $S_1\subseteq \sorts{\Sigma_1}$,
     $T_2$ is a $\Sigma_2$-theory strongly polite w.r.t.\ $S_2\subseteq\sorts{\Sigma_2}$, and
    $S\subseteq S_2$,
then $T_1\oplus T_2$ is strongly polite w.r.t.\ 
$S_1\cup(S_2\setminus S)$.
\end{theorem}


\section{Additive Witnesses}
\label{secweakstrong}
%
It was shown in \cite{JBLPAR} that 
politeness is not sufficient for the proof of the polite combination method
from \cite{10.1007/11559306_3}.  
Strong politeness was introduced to fix the problem.
It is unknown, however, whether there are theories that are polite
but not strongly polite.
In this section we 
offer a simple (yet useful) criterion for the equivalence
of the two notions.
Throughout this section, unless stated otherwise, $\Sigma$ and $S$ denote an
arbitrary signature and a subset of its set of sorts, and
$T,T_{1},T_{2}$ denote arbitrary
$\Sigma$-theories.

The following example, which is based on \cite{JBLPAR} 
using notions of the current paper,
shows that
the strong and non-strong witnesses are different.
Let $\Sigma_{0}$ be a signature with a single sort $\sigma$ 
and no function or predicate symbols (except $=_{\sigma})$, 
 $T_{0}$ the $\Sigma_{0}$-theory
consisting of all $\Sigma_{0}$-structures $\Aa$ with $\card{\sigma^{\Aa}}\geq 2$,
$\phi$ the formula $x=x\land w=w$,
and
$\delta$ the arrangement $(x=w)$ of $\set{x,w}$.
Then $\phi\w\delta$ is $T_{0}$-satisfiable, but every interpretation
$\Aa$ with $\sigma^{\Aa}=\set{x,w}^{\Aa}$ that satisfies it has
only one element in $\sigma^{\Aa}$ and so $\phi$ is not
strongly finitely witnessed for $T_{0}$ w.r.t.\ $\sigma$.
It is straightforward to show, however, that $\phi$ is
finitely witnessed for $T_{0}$ w.r.t.\ $\sigma$.
Moreover, the function $\witness$ defined by
$\witness(\phi)=(\phi\ \w\ w_{1}=w_{1}\ \w\ w_{2}=w_{2})$ for fresh $w_{1},w_{2}$
is a  witness for $T_{0}$ w.r.t.\ $\sigma$, but not a strong one.
This does not show, however, that $T_{0}$ is not strongly polite.
In fact, it is indeed strongly polite 
since the function
$\witness'(\phi)=\phi\w w_{1}\neq w_{2}$ for fresh $w_{1},w_{2}$ is
a strong witness for $T_{0}$ w.r.t.\ $\sigma$.

We introduce 
the notion of additivity, which ensures
that the witness is able to ``absorb" arrangements and thus
lift politeness to strong politeness.

\begin{definition}[Additivity]
\label{additiveDef}
Let $f:\qf(\Sigma)\ra \qf(\Sigma)$.
We say that $f$ is $S$-additive for $T$ if $f(f(\phi)\w\fe)$ and $f(\phi)\w\fe$ 
are $T$-equivalent and have the same set of $S$-sorted variables
for every $\phi,\fe\in \qf(\Sigma)$, provided that 
$\varphi$ is a conjunction of flat literals such that
every term in $\varphi$ is a variable whose sort is in $S$.
When $T$ is clear from the context, we just say that $f$ is $S$-additive.
We say that $T$ is {\em additively} finitely witnessable w.r.t.\ $S$ if there exists a witness for $T$ w.r.t.\ $S$ which is both computable and $S$-additive.
$T$ is said to be {\em additively polite w.r.t.\ $S$} if it is smooth and
additively finitely witnessable w.r.t.\ $S$.
\end{definition}

\begin{proposition}
\label{additive}
Let $\witness$ be a  witness for $T$ w.r.t.\ $S$.
If $\witness$ is $S$-additive 
then it is a strong witness for $T$ w.r.t.\ $S$.\begin{conf}\footnote{Due to lack of space, some proofs have been omitted. They can be found in an extended version at \url{https://arxiv.org/abs/2004.04854}.}\end{conf}\begin{report}\footnote{The omitted proofs can be found in appendix.}\end{report}

\end{proposition}

\begin{corollary}
\label{addfinwit}
Suppose $T$ is additively polite w.r.t.\ $S$. Then it is strongly polite w.r.t.\ $S$.
\end{corollary}

The theory $T_{0}$ from the example above is additively finitely witnessable w.r.t.\ $\sigma$, even though $\witness'$ is not $\sigma$-additive.  
Indeed, it is possible to define a new witness for $T_{0}$ w.r.t.\ $\sigma$, say $\witness''$,  which is $\sigma$-additive. This function $\witness''$ is defined by:
$\witness''(\phi)=
\witness'(\phi)$ if $\phi$ is a conjunction that includes some disequality $x\neq y$
for some $x,y$. Otherwise, $\witness''(\phi)=\phi$.

$T_{0}$ is an {\em existential theory}: it consists of all the structures
that satisfy an existential sentence (in this case, $\exists x,y\:.\:x\neq y$).
The construction of $\witness''$ can be generalized to any existential theory.
Such theories are also smooth w.r.t.\ any set of sorts
and so existential theories are additively polite.

The notion of additive witnesses is useful for proving that a polite theory is strongly polite.
In particular, the witnesses for the theories of equality, arrays, sets and multisets from \cite{10.1007/11559306_3} are all additive, and so
strong politeness of these theories follows from their politeness.
The same will hold later, when we conclude strong politeness of theories
of algebraic datatypes from their politeness.

\section{Politeness for the SMT-LIB 2 Theory of Datatypes}
\label{sec:politedt}

Let $\Sigma$ be a datatypes signature with
$\sorts{\Sigma}=\elemsorts\uplus\structsorts$ and
$\func{\Sigma}=\constructors\uplus\selectors$.
In this section, we prove that $\Ta_{\Sigma}$ is strongly polite with respect to $\elemsorts$. 
In \Cref{sec:inf_IDT}, we consider theories with only inductive
sorts, and consider theories with only finite sorts in \Cref{sec:finite_IDT}. 
We combine them in \Cref{sec:combined_IDT}, where arbitrary theories
of datatypes are considered.
This separation is only needed for finite witnessability.
For smoothness, however, 
it is straightforward to show that the $\elemsorts$ domain
of a given interpretation can always be augmented without
changing satisfiability of quantifier-free formulas.
\begin{lemma}
\label{lem:DT_smooth}
        $\Ta_{\Sigma}$ is smooth w.r.t.\ $\elemsorts$.
\end{lemma}

\Cref{lem:DT_smooth} holds for any datatypes signature.

\subsection{Inductive datatypes}
\label{sec:inf_IDT}

In this section, we assume 
that all sorts in $\structsorts$ are inductive. 

To prove finite witnessability, we now introduce an additive witness function.
Following arguments
from \cite{10.1007/11559306_3}, 
it suffices to define the witness only for conjunctions of flat literals.
A complete witness can then use the restricted one
by first transforming the input formula
to flat DNF form and then creating
a disjunction where each disjunct
is the result of applying the witness on
the corresponding disjunct.
Similarly, it suffices to show that $\witness(\phi)$ is finitely witnessed
for $\phi$ which is a conjunction of flat literals.
Essentially, our witness 
guesses possible constructors for variables whose constructors
are not explicit in the input formula.

\begin{definition}[A Witness for $\Ta_{\Sigma}$]
\label{def:IDT_witness}
        Let $\phi$ be a quantifier-free conjunction of flat $\Sigma$-literals.
        $\iwitness(\phi)$ is obtained from $\phi$ by performing the following steps:
\begin{enumerate}
\item\label{it:selectors} For any literal of the form $y = s_{c,i}(x)$ such
that 
 $x=c(\ora{u_{1}}, y, \ora{u_{2}})$
does not occur in $\phi$ and
$x=d(\ora{u_d})$ does not occur in $\phi$
for any $\ora{u_{1}},\ora{u_{2}},\ora{u_d}$,
we conjunctively add
 $x=c(\ora{u_{1}}, y, \ora{u_{2}}) \lor (\bigvee_{d\neq c}
x=d(\ora{u_d}))$
with fresh 
$\ora{u_{1}},
\ora{u_{2}},
\ora{u_d}$,
where $c$ and $d$ range over $\constructors$.
\item\label{it:tester1} For any literal of the form $is_c(x)$
such that $x=c(\ora{u})$
does not occur in $\phi$ for any $\ora{u}$,
we conjunctively add 
$x=c(\ora{u})$ with fresh $\ora{u}$.
\item\label{it:testers} For any literal of the form $\neg is_c(x)$
such that 
$x=d(\ora{u_{d}})$ does not occur in $\phi$ for
any $d\neq c$ and  $\ora{u_{d}}$,
we conjunctively add
  $\bigvee_{d\neq c}x=d(\ora{u_{d}})$,
with fresh $\ora{u_{d}}$.
\item\label{it:identity} For any sort $\sigma\in \elemsorts$ 
such that $\phi$ does not include a variable of sort $\sigma$
we conjunctively add a literal $x=x$ for a fresh variable $x$ of sort $\sigma$. 
\end{enumerate}
\end{definition}

\begin{example}
Let $\phi$ be the $\siglists$-formula
$y=cdr(x)\w y'=cdr(x)\w is_{cons}(y)$.
$\iwitness(\phi)$ 
is 
$\phi\w (x=nil\vee x=cons(e,y))\w(x=nil\vee x=cons(e',y'))\w y=cons(e'',z)\w e'''=e'''$
where $e,e',e'',e''',z$ are fresh.
\end{example}

In \Cref{def:IDT_witness},
\Cref{it:selectors} guesses the constructor of the argument
for the selector.
\Cref{it:tester1,it:testers} correspond to the semantics of testers.
\Cref{it:identity} is meant to ensure that we can construct
a finite witness with non-empty domains.
The 
requirement for absence of literals before adding
literals or disjunctions to $\phi$
is used to ensure
additivity of $\iwitness$. And indeed:

\begin{lemma}
\label{lem:IDT_additive}
        $\iwitness$ is $\elemsorts$-additive.
\end{lemma}

Further, it can be verified that:
\begin{lemma}
  \label{lem:equivreq}
    Let $\phi$ be a conjunction of flat literals.
    $\phi$ and $\exists \ora{w}\:.\:\Gamma$
    are $\Ta_{\Sigma}$-equivalent,
    where 
  $\Gamma=\iwitness(\phi)$ and
    $\ora{w}=\fv{}{\Gamma}\setminus\fv{}{\phi}$. 
\end{lemma}

The remainder of this section is dedicated to the proof of the
following lemma:

\begin{lemma}[Finite Witnessability]
\label{lem:IDT_witness}
        Let $\phi$ be a conjunction of flat literals.
        Then, $\Gamma=\iwitness(\phi)$ is finitely witnessed for $\Ta_{\Sigma}$ with respect to $\elemsorts$.
\end{lemma}

Suppose that $\Gamma$ is $\Ta_{\Sigma}$-satisfiable, and
let $\Aa$ be a satisfying $\Ta_{\Sigma}$-interpretation.
We define a $\Ta_{\Sigma}$-interpretation $\Ba$ as follows,
and then show that $\Ba$ is a  finite witness of $\Gamma$ for $\Ta_{\Sigma}$ w.r.t.\ $\elemsorts$.
%
First for every $\sigma\in\elemsorts$ we set
$\sigma^{\Ba}=\fv{\sigma}{\Gamma}^\Aa$, and
for every variable $e\in\fv{\sigma}{\Gamma}$, we set
$e^{\Ba}=e^{\Aa}$.
The interpretations of $\structsorts$-sorts,
testers and constructors are uniquely determined by the theory.
It is left to define the interpretation of $\structsorts$-variables in $\Ba$,
as well as the interpretation of the selectors (the 
interpretation of selectors is fixed by the theory only when
applied to the ``right" constructor).
We do this in several steps:

\smallskip
\noindent
{\bf Step 1 -- Simplifying $\Gamma$:}
since $\phi$ is a conjunction of flat literals,
$\Gamma$ is a conjunction
whose conjuncts are either flat literals or
disjunctions of flat literals (introduced
in \Cref{
it:selectors,it:testers} of \Cref{def:IDT_witness}).
Since $\Aa\models\Gamma$, $\Aa$ satisfies exactly one disjunct of each
such disjunction.
We can thus obtain a formula $\Gamma_{1}$
from $\Gamma$ by replacing every 
disjunction with the disjunct that is satisfied by $\Aa$.
Notice that $\Aa\models\Gamma_{1}$ and that it is a conjunction of flat literals.
Let $\Gamma_{2}$ be obtained from $\Gamma_{1}$ by removing any literal
of the form $is_{c}(x)$ and any literal of the form $\neg is_{c}(x)$.
Let $\Gamma_{3}$ be obtained from $\Gamma_{2}$ by removing any literal of the
form $x=s_{c,i}(y)$.  For convenience, we denote $\Gamma_{3}$ by $\Gamma'$.
Obviously, $\Aa\models\Gamma'$, and $\Gamma'$ is a conjunction of flat literals
without selectors and testers.

\smallskip
\noindent
{\bf Step 2 -- Working with Equivalence Classes:}
We would like to preserve equalities between $\structsorts$-variables
from $\Aa$.
To this end, we group all variables in $\fv{}{\Gamma}$ to equivalence classes
according to their interpretation in $\Aa$.
Let    
$\equiv_{\Aa}$ denote an equivalence relation over $\fv{}{\Gamma}$ such that
    $x\equiv_{\Aa}y$ iff $x^{\Aa}=y^{\Aa}$.
    We denote by $[x]$ the equivalence 
    class of $x$. 
Let $\alpha$ be an equivalence class, thus
$\alpha^{\Aa}=\set{x^{\Aa}\mid x\in\alpha}$ is a singleton.
Identifying this singleton with its only element, we have that
$\alpha^{\Aa}$ denotes $a^{\Aa}$ for 
an arbitrary element
$a$ of the equivalence class $\alpha$.

\smallskip
\noindent
{\bf Step 3 -- Ordering Equivalence Classes:}
We would also like to preserve disequalities between $\structsorts$-variables
from $\Aa$. Thus we introduce a relation 
     $\prec$ over the equivalence classes,
such that
    $\alpha \prec \beta$
    if $y = c(w_1\til w_n)$ 
    occurs as one of the conjuncts in $\Gamma'$ for some 
    $w_{1}\til w_n$ and $c$ such that
    $w_k \in \alpha$ for some $y \in \beta$,  $c\in\constructors$, and $k$.
    Call an equivalence class $\alpha$ {\em nullary} if
    $\Aa\models is_c(x)$ for some $x\in \alpha$ and nullary constructor $c$.
    Call an equivalence class $\alpha$ {\em minimal} 
    if $\beta \not\prec \alpha$ for every $\beta$.
Notice that each nullary equivalence class is minimal.
The relation $\prec$ induces a directed acyclic graph (DAG), denoted $G$.
The vertices are the equivalence classes. Whenever $\alpha \prec \beta$, we draw an edge from vertex $\alpha$ to $\beta$.
%
%
%

\smallskip
\noindent
{\bf Step 4 -- Interpretation of Equivalence Classes:}
We define $\alpha^{\Ba}$ for every
equivalence class $\alpha$.
Then, $x^{\Ba}$ is simply defined as $[x]^{\Ba}$, for every
$\structsorts$-variable $x$.
The idea goes as follows. 
Nullary classes are assigned according to $\Aa$.
Other minimal classes are assigned arbitrarily, 
but it is important to assign different classes
to terms whose depths are far enough from each other to
ensure that the disequalities in $\Aa$ are preserved.
Non-minimal classes are uniquely determined after minimal ones are
assigned.
Formally,
let $m$ be the number of equivalence classes, 
$l$ the number of minimal equivalence classes,
$r$ the number of nullary equivalence classes,
and
$\alpha_{1}\til \alpha_{m}$ a topological sort of $G$,
such that all minimal classes occur before all others, 
and the
first $r$ classes are nullary.
Let
$d$ be the length of the longest path in $G$.
We define $\alpha_{i}^{\Ba}$ by induction on $i$.
In the definition,
we use $\Ba_{\elemsorts}$ to denote
the $\elemsorts$-sorted
set assigning $\sigma^{\Ba}$ to every $\sigma\in\elemsorts$.
\begin{enumerate}
\item\label{it:nullary} If $0<r$ and $i\leq r$
then $\alpha_{i}$ is a nullary class and so we set
$\alpha_i^{\Ba}=\alpha_i^{\Aa}$.


\item\label{it:minimal_rest} 
If $r<i\leq l$ then $\alpha_{i}$ is minimal and not nullary.
Let $\sigma$ be the sort of variables in $\alpha_{i}$.
If
$\sigma\in\elemsorts$, then all variables in the class
have already been defined.
Otherwise, $\sigma\in\structsorts$.
In this case,
 we define $\alpha_{i}^{\Ba}$ 
to be an arbitrary element of 
$T_{\sigma}(\consig{\Sigma}, \Ba_{\elemsorts})$ that has
depth strictly greater than 
$\max\set{\depth(\alpha_{j}^{\Ba})\mid 0<j<i}+d$
(here $\max\emptyset=0$).

\item\label{it:non-minimal}If $i>l$ then we set $\alpha_{i}^{\Ba}=c(\beta_{1}^{\Ba}\til\beta_{n}^{\Ba})$
for the unique equivalence classes $\beta_{1}\til \beta_{n}\suq\set{\alpha_{1}\til \alpha_{i-1}}$
and $c$ such that 
$y=c(x_{1}\til x_{n})$ occurs in $\Gamma'$ for some
$y\in\alpha_{i}$ and $x_{1}\in\beta_{1}\til x_{n}\in\beta_{n}$.

\end{enumerate}
Since $\Sigma$ is a datatypes signature in which all $\structsorts$-sorts are inductive, the second case of the definition is
well-defined.
Further,
the topological sort ensures $\beta_{1}\til \beta_{n}$ exist, and the partition
to equivalence classes ensures that they are unique.
Hence:

\begin{lemma}
\label{lem:IDT_well_defined}
$\alpha_{i}^{\Ba}$  is well-defined.
\end{lemma}

\noindent
{\bf Step 5 -- Interpretation of Selectors: }
Let $s_{c,i}\in\selectors$ for
$c:\sigma_{1}\times\ldots\times\sigma_{n}\ra\sigma$,
$1\leq i\leq n$ and $a\in\sigma^{\Ba}$.
If $a\in is_{c}^{\Ba}$, we must have
$a=c(a_{1}\til a_{n})$ for some $a_{1}\in\sigma_{1}^{\Ba}\til a_{n}\in\sigma_{n}^{\Ba}$.
We then set $s_{c,i}^{\Ba}(a)=a_{i}$.
Otherwise, we consider two cases.
If $x^{\Ba}=a$ for some $x\in\fv{}{\Gamma}$ such that
$y=s_{c,i}(x)$ occurs in $\Gamma_{2}$ for some $y$,
we set
$s_{c,i}^{\Ba}(a)=y^{\Ba}$.
Otherwise, 
$s_{c,i}^{\Ba}(a)$ is set arbitrarily.

\begin{example}

Let $\Gamma$ be 
the following $\siglists$-formula:
$x_1 = cons(e_1, x_2) \land x_3 = cons(e_2, x_4) \land x_2 \ne x_4$.
Then $\Gamma'=\Gamma$.
We have the following satisfying interpretation $\Aa$: 
$\elemsort^\Aa = \{1, 2, 3, 4\}$,
$e_1^\Aa = 1, e_2 ^\Aa = 2$, $x_1^\Aa = [1, 2, 3], x_2^\Aa = [2, 3], x_3^\Aa = [2, 2, 4], x_4^\Aa = [2, 4]$.
The construction above yields the following interpretation $\Ba$:
$\elemsort^\Ba = \{1, 2\}$,
$e_1^\Ba = 1, e_2 ^\Ba = 2$. 
For $\listsort$-variables, we proceed as follows.
The equivalence classes of $\listsort$-variables are $[x_1], [x_2], [x_3], [x_4]$, with
$[x_2]\prec [x_1]$ and
$[x_4]\prec [x_3]$.
The length of the longest path in $G$ is 1.
Assuming $[x_2]$ comes before $[x_4]$ in the topological sort,
$x_2^{\Ba}$ will get an arbitrary list over $\set{1,2}$ with length greater than 1 
(the depth of $e_2^{\Ba}$ plus the length of the longest path),
say, $[1,1,1]$.
$x_{4}^{\Ba}$ will then get an arbitrary list of length greater than
$4$ (the depth of $x_2^{\Ba}$ plus the length of the longest path).
Thus we could have $x_{4}^{\Ba}=[1,1,1,1,1]$.
Then, $x_{1}^{\Ba}=[1,1,1,1]$ and
$x_{3}^{\Ba}=[2,1,1,1,1,1]$.
\end{example}

\medskip
Now that $\Ba$ is defined, it is left to show that it is a finite
witness of $\Gamma$ for $\Ta_{\Sigma}$ w.r.t.\ $\elemsorts$.
By construction, $\sigma^{\Ba}=\fv{\sigma}{\Gamma}^{\Ba}$ 
for every $\sigma\in\elemsorts$. 
$\Ba$ also preserves the equalities and disequalities in $\Aa$,
and by considering 
every shape of
a literal in $\Gamma'$ we can prove 
that $\Ba\models\Gamma'$. 
Our interpretation of the selectors then ensures that:

\begin{lemma}
\label{lem:IDT_sat}
    $\Ba \models \Gamma$.
\end{lemma}

\Cref{lem:IDT_sat}, together with the definition of the domains
of $\Ba$, gives us that 
$\Ba$ is a finite witness of 
$\Gamma$ for 
$\Ta_{\Sigma}$ 
w.r.t.\ $\elemsorts$, and so \Cref{lem:IDT_witness} is proven.
As a corollary of \Cref{lem:DT_smooth,lem:IDT_witness,lem:IDT_additive}, strong politeness is obtained.
\begin{theorem}
\label{thm:polite_ind}
If $\Sigma$ is a datatypes signature and all sorts in $\structsorts_{\Sigma}$
are inductive, then
    $\Ta_{\Sigma}$ is strongly polite w.r.t.\ $\elemsorts_{\Sigma}$.
\end{theorem}

\subsection{Finite datatypes}
\label{sec:finite_IDT}
In this section, we assume 
that all sorts in $\structsorts$ are finite. 
%

For finite witnessability, we define the following witness,
that guesses the construction of each $\structsorts$-variables 
until a fixpoint is reached.
For every quantifier-free conjunction of flat $\Sigma$-literals $\phi$,
define the sequence $\phi_{0},\phi_{1},\ldots$, such that
$\phi_{0}=\phi$, and for every $i\geq 0$,
$\phi_{i+1}$ is obtained from $\phi_{i}$ by conjuncting it with
a disjunction
$\bigvee_{c\in\constructors} x=c(w_1^{c}\til w_{n_{c}}^{c})$
for fresh $w_{1}^{c}\til w_{n_{c}}^{c}$,
where $x$ is some arbitrary $\structsorts$-variable in $\phi_{i}$
such that there is no literal of the form
$x=c(y_{1}\til y_{n})$ in $\phi_{i}$ for any constructor $c$
and variables $y_{1}\til y_{n}$,
if such $x$ exists.
Since $\structsorts$ only has finite sorts, this sequence becomes constant
at some $\phi_{k}$.

\begin{definition}[A Witness for $\Ta_{\Sigma}$]
\label{def:IDT_finite_witness}
    $\fwitness(\phi)$ is $\phi_{k}$ for the minimal $k$
such that $\phi_{k}=\phi_{k+1}$.
\end{definition}

\begin{example}
Let $\phi$ be the $\sigpairs$-formula
$x=first(y)\w x'=first(y')\w x\neq x'$.
$\fwitness(\phi)$ 
is
$\phi\w y=pair(e_1,e_2)\w y'=pair(e_3,e_4)$.
\end{example}

Similarly to \Cref{sec:inf_IDT}, we have:

\begin{lemma}
\label{lem:IDT_finite_additive}
        $\fwitness$ is $\elemsorts$-additive.
\end{lemma}

\begin{lemma}
\label{lem:fin-eq}
    $\phi$ and $\exists \ora{w}\:.\:\fwitness(\phi)$
    are $\Ta_{\Sigma}$-equivalent,
    where 
$\ora{w}=\fv{}{\fwitness(\phi)}\setminus\fv{}{\phi}$. 
\end{lemma}

We now prove the following lemma:

\begin{lemma}[Finite Witnessability]
\label{lem:IDT_finite_witness}
        Let $\phi$ be a conjunction of flat literals.
        Then, $\fwitness(\phi)$ is finitely witnessed for $\Ta_{\Sigma}$ with respect to $\elemsorts$.
\end{lemma}

Suppose $\Gamma=\fwitness(\phi)$ is $\Ta_{\Sigma}$-satisfiable, and
let $\Aa$ be a satisfying $\Ta_{\Sigma}$-interpretation.
We define a $\Ta_{\Sigma}$-interpretation $\Ba$ which is
a finite witness of $\Gamma$ for $\Ta_{\Sigma}$ w.r.t.\ $\elemsorts$.
We set
$\sigma^{\Ba}=\fv{\sigma}{\Gamma}^\Aa$ 
for every $\sigma\in\elemsorts$,
 $e^{\Ba}=e^{\Aa}$,
for every variable $e\in\fv{\elemsorts}{\Gamma}$
and
$x^{\Ba}=x^{\Aa}$
for every variable $x\in\fv{\structsorts}{\Gamma}$.
Selectors are also interpreted as they are interpreted in $\Aa$.
This is well-defined: for any $\structsorts$-variable $x$, every element in
$\sigma^{\Aa}$ for $\sigma\in\elemsorts$ that occurs in
$x^{\Aa}$ has a corresponding variable $e$ in $\Gamma$
such that $e^{\Aa}$ is that element.
This holds by the finiteness of the sorts in $\structsorts$ and the
definition of $\fwitness$.
Further, for any $\structsorts$-variable $x$ such that
$s_{c,i}(x)$ occurs in $\Gamma$, 
we must have that it occurs in some literal of the form
$y=s_{c,i}(x)$ of $\Gamma$.
Similarly to the above, all elements that occur in $y^{\Aa}$ and $x^{\Aa}$
have corresponding variables in $\Gamma$.
Therefore, $\Ba\models \Gamma$ is a trivial consequence of $\Aa\models \Gamma$. 
By the definition of its domains, $\Ba$ is a finite witness of 
$\Gamma$ for $\Ta_{\Sigma}$ w.r.t.\ $\elemsorts$, and so \Cref{lem:IDT_finite_witness} is proven.
Then, by \Cref{lem:DT_smooth,lem:IDT_finite_witness,lem:IDT_finite_additive}, strong politeness is obtained.

\begin{theorem}
\label{thm:polite_fin}
If $\Sigma$ is a datatypes signature and all sorts in $\structsorts_{\Sigma}$
are finite, then
    $\Ta_{\Sigma}$ is strongly polite w.r.t.\ $\elemsorts_{\Sigma}$.
\end{theorem}


\subsection{Combining finite and inductive datatypes}
\label{sec:combined_IDT}

Now we consider the general case.
Let $\Sigma$ be a datatypes signature.
We prove that $\Ta_{\Sigma}$ is strongly polite w.r.t.\ $\elemsorts$.
We show that 
there are datatypes signatures $\Sigma_{1},\Sigma_{2}\suq\Sigma$
such that 
$\Ta_{\Sigma}=\Ta_{\Sigma_{1}}\oplus\Ta_{\Sigma_{2}}$,
and then use \Cref{thm:lpar}.
In $\Sigma_{1}$, inductive sorts are excluded, while in
$\Sigma_{2}$, finite sorts are considered to be element sorts.

Formally, 
we set $\Sigma_{1}$ as follows:
where $\elemsorts_{\Sigma_{1}}=\elemsorts_{\Sigma}$
and $\structsorts_{\Sigma_{1}}=\Fin(\Sigma)$.
$\func{\Sigma_{1}}=\constructors_{\Sigma_{1}}\uplus\selectors_{\Sigma_{1}}$,
where
$\constructors_{\Sigma_{1}}=\{c:\sigma_{1}\times\ldots\times\sigma_{n}\ra\sigma\mid c\in\constructors_{\Sigma},\sigma\in
\structsorts_{\Sigma_{1}}\}$
and
$\selectors_{\Sigma_{1}}$
and $\predecessors_{\Sigma_{1}}$ are
the corresponding selectors and testers.
Notice that if $\sigma$ is finite and 
$c:\sigma_{1}\times\ldots\times \sigma_{n}\ra\sigma$ is in
$\constructors_{\Sigma}$,
then $\sigma_{i}$ must be finite or in $\elemsorts_{\Sigma}$ for every $1\leq i\leq n$.
Next, we set $\Sigma_{2}$ as follows: $\sorts{\Sigma_{2}}=\elemsorts_{\Sigma_{2}}\uplus\structsorts_{\Sigma_{2}}$,
where $\elemsorts_{\Sigma_{2}}=\elemsorts_{\Sigma}\cup \Fin(\Sigma)$
and $\structsorts_{\Sigma_{2}}=\Ind(\Sigma)$.
$\func{\Sigma_{2}}=\constructors_{\Sigma_{2}}\uplus\selectors_{\Sigma_{2}}$,
where
$\constructors_{\Sigma_{2}}=\{c:\sigma_{2}\times\ldots\times\sigma_{n}\ra\sigma\mid
c\in\constructors_{\Sigma},\sigma\in \structsorts_{\Sigma_{2}}\}$
and
$\selectors_{\Sigma_{2}}$ and $\predecessors{\Sigma_{2}}$ are
the corresponding selectors and testers.
Thus,
$\Ta_{\Sigma}=\Ta_{\Sigma_{1}}\oplus\Ta_{\Sigma_{2}}$.
Now set $S=\elemsorts_{\Sigma}\cup \Fin(\Sigma)$,
$S_{1}=\elemsorts_{\Sigma}$, 
$S_{2}=\elemsorts_{\Sigma}\cup \Fin(\Sigma)$,
$T_{1}=\Ta_{\Sigma_{1}}$, and
$T_{2}=\Ta_{\Sigma_{2}}$.

By \Cref{thm:polite_fin}, $T_{1}$ is strongly polite w.r.t.\ $S_{1}$
and by \Cref{thm:polite_ind}, $T_{2}$ is strongly polite w.r.t.\ $S_{2}$.
By \Cref{thm:lpar}
we have:

\begin{theorem}
\label{thm:polite}
If $\Sigma$ is a datatypes signature 
then
    $\Ta_{\Sigma}$ is strongly polite w.r.t.\ $\elemsorts_{\Sigma}$.
\end{theorem}

\begin{remark}
\label{rem:composition}
A concrete witness for $\Ta_{\Sigma}$ in the general case,
that we call $\witness_{\Sigma}$,
is obtained by first applying the witness from \Cref{def:IDT_witness}
and then applying the witness from \Cref{def:IDT_finite_witness} on the literals
that involve finite sorts.
A direct finite witnessability proof can be obtained 
by using the same arguments from the proofs of
\Cref{lem:IDT_witness,lem:IDT_finite_witness}.
This witness is simpler than the one produced in the proof
from \cite{JBLPAR} of \Cref{thm:lpar}, that involves purification and 
arrangements. In our case, we do not consider arrangements, but instead
notice that the resulting function is additive, and hence
ensures
strong finite witnessability.
\end{remark}

\section{Axiomatizations}
\label{sec:axiomatizations}
In this section, we discuss the possible connections between
the politeness of $\Ta_{\Sigma}$ and some axiomatizations of trees. 
We show how
 to get a
reduction of any $\Ta_{\Sigma}$-satisfiability problem into a 
satisfiability problem modulo an axiomatized theory of trees.
The latter can be decided using syntactic unification. 

Let $\Sigma$  be a datatypes signature. 
The set 
$\TREE_{\Sigma}^{\ast}$ of axioms is defined as the union of all the sets of axioms in \Cref{fig:axioms} (where upper case letters denote implicitly universally quantified variables).
Let $\TREE_{\Sigma}$ be the set obtained from
$\TREE_{\Sigma}^{\ast}$ by dismissing
${\mathit{Ext}_{1}}$ and
${\mathit{Ext}_{2}}$.
Note that because of 
$\mathit{Acyc}$, we have that
$\TREE_{\Sigma}$ is infinite
(that is, consists of infinitely many axioms)
unless all sorts in $\structsorts$ are finite.
$\TREE_{\Sigma}$ is a generalization of the 
theory of Absolutely Free Data Structures (AFDS)
from \cite{CFR19}
to many-sorted signatures with selectors and testers.
In what follows we identify $\TREE_{\Sigma}$ (and $\TREE_{\Sigma}^{\ast}$)
with the class of structures that satisfy them when there is no ambiguity.

\begin{proposition}
\label{sound}
Every $\TREE_{\Sigma}^{\ast}$-unsatisfiable formula is 
$\Ta_{\Sigma}$-unsatisfiable.
\end{proposition}

\begin{figure}
 \centering

\fbox{
 
\begin{minipage}{115mm}
$$
\begin{array}{l@{~~~~~}l}
(\mathit{Inj}) & \{ c(X_1,\dots,X_n) = c(Y_1,\dots,Y_n) \ra \bigwedge_{i=1}^n X_i = Y_i ~|~ c \in \constructors \}\\
(\mathit{Dis}) & \{ c(X_1,\dots,X_n) \neq d(Y_1,\dots,Y_m) ~|~ c,d \in \constructors, c \neq d \}\\
(\mathit{Proj}) & \{ s_{c,i}(c(X_1, \dots, X_n)) = X_i ~|~ c \in \constructors, i \in [1,n] \} \\
({\Is_{1}}) & \{ is_c(c(X_1, \dots, X_n)) ~|~ c \in \constructors \} \\
({\Is_{2}}) & \{ \neg is_c(d(X_1, \dots, X_n)) ~|~ c,d \in \constructors, c \neq d \} \\
(\mathit{Acyc}) & \{ X \neq t[X] ~|~ \mbox{$t$ is a non-variable $\consig{\Sigma}$-term that contains $X$ } \} \\
({\mathit{Ext}_{1}}) & 
\{ \bigvee_{c: 
\sigma_{1}\times\ldots\times\sigma_{n}\ra\sigma \in 
\constructors} is_{c}(X) ~|~ \sigma \in \structsorts\} \\
({\mathit{Ext}_{2}}) & 
\{ \exists \ora{y}\:.\:is_c(X) \ra X = c(\ora{y}) ~|~ c \in \constructors \}
\end{array}
$$
\end{minipage}
}
\caption{Axioms for $\TREE_{\Sigma}$ and $\TREE_{\Sigma}^{\ast}$}
 \label{fig:axioms}
\end{figure}

\begin{remark}
\label{superpositionremark}
Along the lines of \cite{DBLP:journals/tocl/ArmandoBRS09}, a superposition calculus can be applied to get a $\TREE_{\Sigma}$-satisfiability procedure. Such a calculus has been used in \cite{DBLP:journals/entcs/BonacinaE07,CFR19} for a theory of trees with selectors but no testers. To handle testers, one can use a classical encoding of predicates into first-order logic with equality, by representing an atom $is_c(x)$ as a flat equality $\Is_c(x) = \mathbb{T}$ where $\Is_c$ is now a unary function symbol and $\mathbb{T}$ is a constant. Then, a superposition calculus dedicated to $\TREE_{\Sigma}$ can be obtained by extending the standard superposition calculus \cite{DBLP:journals/tocl/ArmandoBRS09} with some expansion rules, one for each axiom of $\TREE_{\Sigma}$ \cite{CFR19}. For the axioms ${\Is_{1}}$ and ${\Is_{2}}$, the corresponding expansion rules are respectively
$
x = c(x_1,\dots,x_n) \vdash \Is_c(x) = \mathbb{T} \mbox{~ if $c\in \constructors$}$, and 
$x = d(x_1,\dots,x_n) \vdash \Is_c(x) \neq \mathbb{T} \mbox{~ if $c,d \in \constructors, c \neq d$}$.
Further,
consider the theory of finite trees defined from $\TREE_{\Sigma}$ by dismissing $\mathit{Proj}, {\Is_{1}}$ and ${\Is_{2}}$. 
Being defined by Horn clauses, it is convex.
Further, it is a Shostak theory~\cite{DBLP:journals/jacm/Shostak79,DBLP:conf/unu/MannaZ02,DBLP:journals/iandc/KrsticC05} admitting a solver and a canonizer \cite{CFR19}. The solver is given by a syntactic unification algorithm \cite{DBLP:books/el/RV01/BaaderS01} and the canonizer is the identity function. The satisfiability procedure built using the solver and the canonizer can be applied to decide $\TREE_{\Sigma}$-satisfiability problems containing $\consig{\Sigma}$-atoms.
\end{remark}

The following result shows that any $\Ta_{\Sigma}$-satisfiability problem can 
be reduced to a $\TREE_{\Sigma}$-satisfiability problem. 
This leads to a $\Ta_{\Sigma}$-satisfiability procedure. 

\begin{proposition}
\label{groundcomplete}
Let $\Sigma$ be a finite datatypes signature
and $\varphi$ any conjunction of flat $\Sigma$-literals
including an arrangement over the variables in $\varphi$. 
Then, there exists a $\Sigma$-formula $\varphi'$
such that:
\begin{enumerate}
\item $\varphi$ and $\exists\ora{w}\:.\:\varphi'$ are $\Ta_{\Sigma}$-equivalent, where
$\ora{w}=\fv{}{\varphi'} \backslash \fv{}{\varphi}$.
\item $\varphi'$ is $\Ta_{\Sigma}$-satisfiable iff $\varphi'$ is $\TREE_{\Sigma}$-satisfiable.
\end{enumerate}
\end{proposition}

\Cref{groundcomplete} can be easily lifted to any conjunction of
$\Sigma$-literals $\varphi$ by flattening and then guessing all
possible arrangements over the variables.
%
Further, $\exists\ora{w}\:.\:\varphi'$ 
and $\varphi$ are not only $\Ta_{\Sigma}$-equivalent
but also $\TREE_{\Sigma}^{\ast}$-equivalent.
As a consequence,
Proposition~\ref{groundcomplete} also holds when 
stated using $\TREE_{\Sigma}^{\ast}$ instead of
$\Ta_{\Sigma}$. 

\medskip
We conclude this section with a short discussion on the connection to
\Cref{sec:politedt}.
Both the current section 
and \Cref{sec:politedt} rely on two constructions: 
$(i)$~A formula transformation ($\witness_{\Sigma}$ in \Cref{sec:politedt},
$\varphi\mapsto\varphi'$ in the current section);
and $(ii)$~A small model construction (finite witnessability in \Cref{sec:politedt},
equisatisfiability between $\Ta_{\Sigma}$ and $\TREE$ in \Cref{groundcomplete}).
While these constructions are similar in both sections, they are not the same.
A nice feature of the constructions of \Cref{sec:politedt} is that they
clearly separate between steps $(i)$ and $(ii)$.
The witness is very simple, and amounts to adding to the input formula
literals and disjunctions that trivially follow from
the original formula in $\Ta_{\Sigma}$.
Then, the resulting formula is post-processed in step $(ii)$,
according to a given satisfying interpretation.
Having a satisfying interpretation allows us to greatly simplify the formula,
and the simplified formula is useful for the model construction.
In contrast, the satisfying $\TREE_{\Sigma}$-interpretation that we start with in step $(ii)$
of the current section is
not necessarily a $\Ta_{\Sigma}$-interpretation,
which makes the approach of \Cref{sec:politedt} incompatible, compared
to the syntactic unification approach that we employ here.
For that, some of the post-processing steps of \Cref{sec:politedt} are employed in 
step $(i)$ itself, in order to eliminate all testers and as much selectors
as possible.
In addition, a pre-processing is applied in order to include an arrangement.
The constructed interpretation finitely witnesses $\varphi'$ and
so this technique can be used to produce an alternative proof of politeness.

\section{Conclusion}
\label{sec:conc}
In this paper we have studied the theory of algebraic
datatypes, as it is defined by the \smtlib standard.
Our investigation included both finite and inductive datatypes.
For this theory, we have proved that it is strongly polite,
making it amenable for combination with other theories by
the polite combination method.
Our proofs used the notion of additive witnesses, also 
introduced in this paper.
We concluded by extending existing axiomatizations and
a decision procedure of trees to support this theory of datatypes.

There are several directions for further research that 
we plan to explore.
First, we plan to continue to prove that more important theories are
strongly polite, with an eye to recent extensions of the datatypes theory,
namely datatypes with shared selectors \cite{DBLP:conf/cade/ReynoldsVBTB18} 
and co-datatypes \cite{DBLP:journals/jar/ReynoldsB17}.
Second, we envision to further investigate the possibility to prove politeness using superposition-based satisfiability procedures. 
%
Third, we plan to study extensions of the theory of datatypes
corresponding to finite trees including function symbols with some
equational properties such as associativity and commutativity to model data structures such as
multisets \cite{DBLP:conf/cade/Sofronie-Stokkermans09}. We want to
focus on the politeness of such extensions.  Initial work in that
direction has been done in \cite{SMT16}, that we plan to build on.

\paragraph{Acknowledgments.}  We are thankful to the anonymous reviewers for their comments.

\bibliographystyle{splncs04}
\bibliography{polite}

\begin{report}
  
\newpage\appendix

\section{Proofs}
In the proofs that follow, we use the following notations.
When $X$ is a set of $\Sigma$-sentences and $Y$ a set of $\Sigma$-formulas, we write
$X\consequence{\Sigma}Y$ when every $\Sigma$-interpretation
that satisfies all the formulas in $X$ must satisfy all the formulas in $Y$.
When $X$ is a set of $\Sigma$-formulas, we write
We write $\valid{T}X$ to state the every formula in $X$ is satisfied
by every $T$-interpretation.
We sometimes identify conjunctions of literals
with sets of literals, when there is no ambiguity.

\subsection{Proof of \Cref{additive}}
Let $\phi\in \qf(\Sigma)$. 
We prove that $\witness(\phi)$ is strongly finitely witnessed for $T$ w.r.t.\ $S$.
Let $V$ be a set of variables of sorts in $S$ and $\delta_{V}$ an arrangement of $V$.
We prove that $\witness(\phi)\w\delta_{V}$ is  finitely witnessed for $T$ w.r.t.\ $S$. 
Suppose it is $T$-satisfiable. Then since $\witness$ is $S$-additive
 and $\delta_{V}$
is a conjunction of flat literals that 
contains only variables of sorts in $S$ as terms,
$\witness(\witness(\phi)\w\delta_{V})$ is also $T$-satisfiable.
$\witness$ is a witness for $T$ w.r.t.\ $S$, and hence
$\witness(\witness(\phi)\w\delta_{V})$ has a finite witness $\Aa$ for $T$
w.r.t.\ $S$.
By $T$-equivalence, $\Aa\models\witness(\phi)\w\delta_{V}$.
Since both formulas have the same set of $S$-variables,
$\Aa$ is also a finite witness of $\witness(\phi)\w\delta_{V}$.
\qed

\subsection{Existential Theories are Strongly Polite}
To complement the discussion of 
\Cref{secweakstrong} on existential theories, we prove the following
proposition:

\begin{proposition}
\label{prop:existext}
\label{conc:exist}
If $T$ is existential then it is strongly polite w.r.t.\ $S$.
\end{proposition}

\begin{proof}
Let $\varphi$ be the formula whose existential closure defines
$T$.
Define a function
$\witness_{T}$ by
\[\witness_{T}(\phi)=
\begin{cases}
\phi & \phi=\fe_{1}\w\ldots\w \varphi'\w\ldots\w\fe_{n}\\
\phi\w\varphi'' & otherwise
\end{cases}\]
where $\varphi'$ is some formula obtained from $\varphi$ by replacing its variables
with variables not in $\fv{}{\fe_{i}}$ for every $1\leq i\leq n$, and
$\varphi''$ is an any formula obtained from $\varphi$ by replacing its
variables with variables not in $\fv{}{\phi}$.
It can be shown that $\witness_{T}$ is an $S$-additive witness for $T$ w.r.t.\ $S$.

For smoothness,
we note that the construction of 
a satisfying interpretation with adequate cardinalities
can be obtained in a similar fashion to
the construction done in \cite{10.1007/11559306_3} 
(proof of proposition 23).
The resulting interpretation 
is in $T$.
The proposition is then obtained by \Cref{addfinwit}.
\end{proof}

\subsection{Proof of \Cref{lem:IDT_additive}}
As mentioned before \Cref{def:IDT_witness},
$\iwitness$ is extended from conjunctions of flat literals
to arbitrary quantifier-free formulas by transforming
the input formula to flat DNF form and then applying 
the witness on each disjunct of the DNF, taking
the disjunction of these applications.
Full details about the correctness of this process
can be found in \cite{JBLPAR}.
A similar argument can be made for additivity.
We do so here for the case of $\iwitness$ that was defined in \Cref{def:IDT_witness}.

Let $\phi$ be a quantifier-free $\Sigma$-formula,
 $D_{1}\vee\ldots\vee D_{m}$ its flat-DNF
form, and $\varphi$ a conjunction of flat literals
such that every term in $\varphi$ is a variable whose sort
is in $\elemsorts$.
By the above,
$
\iwitness(\iwitness(\phi)\w\varphi)=
\iwitness(\iwitness(D_{1}\vee\ldots\vee D_{m})\w\varphi)=
\iwitness((\iwitness(D_{1})\vee\ldots\vee \iwitness(D_{m}))\w\varphi)
$.
For each $1 \leq i\leq m$, 
let $E_{i}^{1}\vee\ldots E_{i}^{k_{i}}$ be the flat DNF form
of $\iwitness(D_{i})$. 
Since $\iwitness$ does not introduce non-flat literals, no new variables are introduced in the
transformation from $\iwitness(D_{i})$ to
$E_{i}^{1}\vee\ldots \vee E_{i}^{k_{i}}$, but only propositional
transformations are employed.
The equation list above can continue with
$
\iwitness((E_{1}^{1}\w\varphi)\vee \ldots \vee (E_m^{k_{m}}\w\varphi))=
\iwitness(E_{1}^{1}\w\varphi)\vee\ldots\vee
\iwitness(E_m^{k_{m}}\w\varphi)
$.
Now, for each $1\leq i\leq m$ and $1\leq j\leq k_{i}$,
$E_i^{j}$ is a conjunction of flat literals in the DNF-form of
$\iwitness(D_{i})$.
By the construction of $\iwitness$, 
each such $E_i^j\w\varphi$
does not satisfy any of the preconditions in $\iwitness$
for the addition of any formula:
constructors have already been
guessed for selectors and testers, and each conjunction in the DNF
includes a guess. In addition, $\varphi$ does not contain any
constructors and testers. Also, each conjunction in the DNF
includes at least one variable of each $\structsorts$-sort.
Thus $\iwitness(E_i^j\w\varphi)=E_i^j\w\varphi$.
This means that $\iwitness(\iwitness(\phi)\w\varphi)=
(E_{1}^{1}\w\varphi)\vee\ldots\vee (E_m^{k_{m}}\w\varphi)
$.

Similarly, 
$
\iwitness(\phi)\w\varphi=
(\iwitness(D_{1})\vee\ldots \iwitness(D_{m}))\w\varphi
$,
which is logically equivalent to
$(E_{1}^{1}\vee \ldots\vee E_m^{k_{m}})\w\varphi$, and hence to
$(E_{1}^{1}\w\varphi)\vee \ldots\vee (E_m^{k_{m}}\w\varphi)$, which
by the above is equivalent to 
$\iwitness(\iwitness(\phi)\w\varphi)$.
Further, since $\iwitness$ does not introduce non-flat literals,
the set of $\elemsorts$-variables is the same in both formulas.
\qed

\subsection{Proof of \Cref{lem:equivreq}}
Each variable in $\ora{w}$ occurs exactly once in $\Gamma$.
Let $\Gamma'$ be obtained from $\exists\ora{w}.\Gamma$ by pushing
each existential quantifier to the literal that contains
its corresponding quantified variable.
Clearly, $\exists\ora{w}\Gamma$ and $\Gamma'$ are logically equivalent, and in
particular they are $\Ta_{\Sigma}$-equivalent.
$\Gamma'$ contains all the conjuncts of $\phi$ as top-level
conjuncts. Hence clearly every $\Ta_{\Sigma}$-interpretation
that satisfies $\Gamma'$ also satisfies $\phi$.
For the converse, let $\Aa$ be a $\Ta_{\Sigma}$-interpretation that
satisfies $\phi$
and $\Delta$ a top-level conjunct of $\Gamma'$.
\bi
\item If $\Delta$ is also a literal of $\phi$ then $\Aa\models\Delta$.
\item If $\Delta$ corresponds to a formula that was added by \Cref{it:selectors} of \Cref{def:IDT_witness},
then it has the form 
$(\exists \ora{u_{1}} y \ora{u_{2}}.x=c(\ora{u_{1}}, y, \ora{u_{2}})) \lor
(\bigvee_{d\neq c} \exists\ora{u_{d}}.x=d(\ora{u_{d}}))$ and $y=s_{c,i}(x)$ is a
literal of $\phi$.
$\Aa\models y=s_{c,i}(x)$. 
If $\Aa\models is_{c}(x)$ then it must satisfy the first disjunct of $\Delta$.
Otherwise, $\Aa$ must
satisfy one of the other disjuncts. 
In both cases
$\Aa\models\Delta$.
\item If $\Delta$ corresponds to a formula that was added by \Cref{it:tester1} of \Cref{def:IDT_witness}
then it has the form $\exists\ora{u}.x=c(\ora{u})$ and $is_{c}(x)$ is 
a literal of $\phi$.
Since $\Aa\models is_{c}(x)$, we must have
$\Aa\models \Delta$.
\item If $\Delta$ corresponds to a formula that was added by \Cref{it:testers} of \Cref{def:IDT_witness}
then it has the form
$\bigvee_{d\neq c}\exists\ora{u}x=d(\ora{u})$
and $\neg is_{c}(x)$ is in $\phi$.
Since $\Aa\not\models is_{c}(x)$, we must have
$\Aa\models\Delta$.
\item If $\Delta$ corresponds to a formulas that was added by \Cref{it:identity} then it is trivially satisfied.

\ei
\qed

\subsection{Proof of \Cref{lem:IDT_well_defined}}
The case of nullary and minimal constructors is
clearly well-defined.
Suppose $\alpha_{i}$ is not minimal.
Then the sort of its variables is in $\structsorts$.
We prove that there is a unique list 
$\beta_{1}\til \beta_{n}$, of equivalence classes, all elements of
$\set{\alpha_{1}\til \alpha_{i-1}}$ 
 and a unique constructor $c$
such that 
$y=c(x_{1}\til x_{n})$ occurs in $\Gamma'$ for some
$y\in\alpha_{i}$ and $x_{1}\in\beta_{1}\til x_{n}\in\beta_{n}$.
{\bf Existence:} 
$\alpha_{i}$ is not minimal. 
Hence there exists some $\beta_{1}$ such that
$\beta_{1}\prec\alpha_{i}$. 
Hence w.l.g.\ there exists some $y\in\alpha_{i}$ and some
$x_{1}\in\beta_{1}$ such that
$y=c(x_{1},x_{2}\til x_{n})$ is in $\Gamma'$ for some
$x_{2}\til x_{n}$ and $c$.
By definition, this means that $[x_{2}]\til [x_{n}]\prec\alpha_{i}$
as well, and thus $[x_{j}]$ must occur before $\alpha_{i}$
in the topological ordering for every $1\leq j\leq n$, hence
$[x_{j}]\in\set{\alpha_{1}\til \alpha_{i-1}}$ for each $j$.
{\bf Uniqueness:}
Suppose there are also equivalence classes
$\beta_{1}'\til \beta_{m}'$, all elements of
$\set{\alpha_{1}\til \alpha_{i-1}}$ 
 and a constructor $c'$
such that 
$y'=c'(x_{1}'\til x_{m}')$ occurs in $\Gamma'$ for some
$y'\in\alpha_{i}$ and $x_{1}'\in\beta_{1}'\til x_{m}'\in\beta_{m}'$.
Since  $y'=c'(x_{1}'\til x_{m}')$
and
$y=c(x_{1}\til x_{n})$ both occur in $\Gamma'$ and are
thus satisfied by $\Aa$, and $[y]=[y']$,
we must have $c=c'$, $n=m$,
and $\Aa\models x_{j}=x_{j}'$ for every $j$, hence
$[x_{j}]=[x_{j}']$ for every $j$.
\qed

\subsection{Proof of \Cref{lem:IDT_sat}}
We start with the following lemma:

\begin{lemma}
\label{lem:xeqy}
        $x^{\Aa} = y^{\Aa}$ iff $x^{\Ba} = y^{\Ba}$ for every $x, y \in \fv{}{\Gamma}$.
\end{lemma}

\begin{proof}
    The left-to-right direction follows directly from the definition of $\Ba$,
    that does not distinguish distinct elements
    inside a single equivalence class of $\equiv_{\Aa}$.
    For the converse, we prove that 
         $\alpha_1^{\Ba}\til \alpha_p^{\Ba}$ are pairwise distinct
for every $1\leq p\leq m$ by induction on $p$.
From this the claim follows: if $x^{\A}\neq y^{\Aa}$,
then $[x]=\alpha_{p}$ and $[y]=\alpha_{q}$ for some $p\neq q$,
and therefore $x^{\Ba}=[x]^{\Ba}\neq [y]^{\Ba}=y^{\Ba}$.

        Consider the base case for the first $l$ classes.
        {\textcircled{\small 1}}For all the equivalence classes of $\elemsorts$-sorted variables,
as they are also minimal, and the definition is the same as in $\Aa$, their interpretations are distinct by definition.
{\textcircled{\small 2}}For the nullary classes, the definition is also the same as in $\Aa$, thus they have distinct interpretations.
{\textcircled{\small 3}}For the equivalence classes of non-nullary $\structsorts$-sorted variables, 
they have different interpretations with the nullary classes, as their interpretations all have the depth more than $d$.
And among themselves,
the depths of the interpretations
of these classes is a strongly increasing monotonic sequence by definition.

    Now assume the claim for $p$ ($l\leq p<n$) vertices.
    It is sufficient to prove that $\alpha_{p+1}$ has a different interpretation from all the previous vertices.
    Assume otherwise, and let $i\leq p$ with $\alpha_{i}^{\Ba}=\alpha_{p+1}^{\Ba}$.
        $\alpha_{p+1}$ is not minimal.
        Since $\alpha_{p+1}$ cannot be nullary,
        $\alpha_i^{\Ba}=\alpha_{p+1}^{\Ba}$ cannot be nullary, thus we have $i>r$.
        Then let us consider two cases.
        \be
            \item $\alpha_i$ is also not minimal:
            There must be a constructor $c$ such that
                $\alpha_i^{\Ba} = c(\beta_1^{\Ba}\til \beta_{n}^{\Ba})$
and
            $\alpha_{p+1}^{\Ba} = c(\hat{\beta}_1^{\Ba}\til \hat{\beta}_{n}^{\Ba})$
            for some equivalence classes $\beta_1\til \beta_n$ and $\hat{\beta}_1\til \hat{\beta}_n$.
            Then from $\alpha_i^{\Ba}=\alpha_{p+1}^{\Ba}$,
            we have $\beta_k^{\Ba}=\hat{\beta}_k^{\Ba}$ for $k=1\til n$.
            
            Also, note that $\beta_1\til \beta_n, \hat{\beta}_1\til \hat{\beta}_n\in \{\alpha_1\til \alpha_p\}$.
            Let $1\leq k\leq n$.
            By the induction hypothesis, either $d_k=\hat{d}_k$ or $d_k^{\Ba}\neq \hat{d}_k^{\Ba}$.
            By the above, the former must hold.
            Similarly to the proof of  \Cref{lem:IDT_well_defined}, we get that $\alpha_i^{\Aa}=\alpha_{p+1}^{\Aa}$,
            which is a contradiction to the fact that $i<p+1$.
            
            \item $\alpha_i$ is minimal:
            An equivalence class $\beta$ is said to be a source of $\alpha_{p+1}$,
                if there is a path from $\beta$ to $\alpha_{p+1}$ in $G$ and $\beta$ is minimal.
                
                If $\alpha_{p+1}$ has a source vertex $\beta_j$ such that $\depth(\beta_j^{\Ba}) \geq \depth(\alpha_i^{\Ba})$,
then there is $\depth(\alpha_{p+1}^{\Ba}) > \depth(\beta_j^{\Ba})\geq \depth(\alpha_i^{\Ba})$.
      
                Otherwise any source vertex $\beta_j$ of $\alpha_{p+1}$ has $\depth(\beta_j^{\Ba})+d<\depth(\alpha_i^{\Ba})$. 
But the depth of $\alpha_{p+1}^{\Ba}$ is at most $depth(\beta_{j}^{\Ba})+D$ for a source vertex $\beta_j$ which has the highest depth.
Thus $\alpha_i^{\Ba}\neq \alpha_{p+1}^{\Ba}$.
        \ee
\end{proof}

We now proceed with the proof of \Cref{lem:IDT_sat}.
We start by proving that $\Ba\models \Gamma'$.
     $\Gamma'$ is a conjunction of
    flat literals without selectors and testers.
    We consider each type of conjunct separately.
        \bi
            \item Literals of the form $x=y$ or $x\neq y$:
                By Lemma \ref{lem:xeqy}, and the fact that $\Aa\models\Gamma'$,
                these literals hold in interpretation $\Ba$.
        
        \item Literals of the form $x=c$, where $c$ is a nullary constructor:
                In this case, $x^{\Ba}$ is defined as $x^{\Aa}$.
                Since $\Aa \models \Gamma'$, we have $\Ba\models x=c$.

            \item Literals of the form $x=c(w_1\til w_{n})$ for some constructor $c$:
                Since $\Aa \models \Gamma'$, 
                $c$ is the only constructor that construct $x$ in $\Gamma'$.
                From the definition of $\Ba$,
                $x^{\Ba}=c(d_1^{\Ba}\til d_n^{\Ba})$ for some $d_1\til d_n$.
                And by \Cref{lem:IDT_well_defined},
                we have $[w_k]=[d_k]$ for $k=1\til n$.
                So we have $x^{\Ba}=c(w_1^{\Ba}\til w_n^{\Ba})$.
        \ei

Next, we prove that $\Ba\models\Gamma_{2}$.
$\Gamma_{2}$ is a conjunction of the literals of $\Gamma'$, together
with literals of the form $y=s_{c,i}(x)$ from $\Gamma$.
Let $y=s_{c,i}(x)$ be such a conjunct of $\Gamma_{2}$.
Then by the definition of $\iwitness$ and $\Gamma'$,
there are two cases:
\bi
\item $x=c(\ldots,y,\ldots)$ is in $\Gamma'$.
Thus $[y]\prec [x]$ and 
$x^{\Ba}=c(\ldots,y^{\Ba},\ldots)$ by
the definition of $\Ba$.
In particular, $x^{\Ba}\in is_{c}^{\Ba}$.
In this case, $s_{c,i}(x)^{\Ba}$ is set to $y^{\Ba}$
by the definition of $\Ba$.
\item $x=d(\ldots)$ is in $\Gamma'$ for some $d\neq c$.
We consider the following sub-cases.
\bi
\item If $d$ is nullary then $[x]$ is nullary.
In this case, $x^{\Ba}=x^{\Aa}$.
$\Aa\models\Gamma'$ and hence $x^{\Aa}\in is_{d}^{\Aa}$,
which means that
$x^{\Ba}\in is_{d}^{\Ba}$ as well.
In particular, $x^{\Ba}\notin is_{c}^{\Ba}$.
Since $y=s_{c,i}(x)$ occurs in $\Gamma_{2}$, 
$s_{c,i}(x)^{\Ba}$ is set to be $y^{\Ba}$.
\item If $d$ is not nullary then 
$[x]$ cannot be minimal, and hence $x^{\Ba}\in is_{d}^{\Ba}$ by the definition of $\Ba$.
In particular, $x^{\Ba}\notin is_{c}^{\Ba}$.
Since $y=s_{c,i}(x)$ occurs in $\Gamma_{2}$,
$s_{c,i}(x)^{\Ba}$ is set to be $y^{\Ba}$ in this case.
\ei
\ei
Hence $\Ba\models \Gamma_{2}$.

Next, we show that $\Ba\models\Gamma_{1}$, which is obtained from $\Gamma_{2}$
by the addition of conjunctions of the form $is_{c}(x)$ and $\neg is_{c}(x)$.
Let $is_{c}(x)$ be such a literal in $\Gamma_{1}$.
Then it is also a literal of $\Gamma$.
Then by the definition of $\iwitness$ and of $\Gamma'$, 
this means that $\Gamma'$ contains a literal of the form
$x=c(y_{1}\til y_{n})$.
Since $\Ba\models\Gamma'$, we have $\Ba\models is_{c}(x)$.
Now let $\neg is_{c}(x)$ be a literal of $\Gamma_{1}$.
Then it is also a literal of $\Gamma$.
By the definition of $\iwitness$ and $\Gamma'$,
the latter contains a literal of the form 
$x=d(t_{1}\til y_{n})$ for some $d\neq c$.
Since $\Ba\models\Gamma'$, we have
$\Ba\models\neg is_{c}(x)$.

Finally, we have seen that
$\Ba$ satisfies a disjunct in every
disjunction of $\Gamma$, as well as all of the top-level
literals of $\Gamma$, which means that
$\Ba\models\Gamma$.
\qed

  \subsection{Proof of \Cref{groundcomplete}}

We first recall some standard notions of (syntactic) unification~\cite{DBLP:books/el/RV01/BaaderS01}.
  
Given a signature $\Sigma$ and a denumerable set of variables $V$, the set of
$\Sigma$-terms over $V$ defines a $\Sigma$-structure, also denoted by
$T(\Sigma,V)$.  A {\em substitution} is an endomorphism of $T(\Sigma,V)$
with only finitely many variables not mapped to themselves. A
substitution is denoted here by $\sigma = \{ x_k \mapsto t_k \}_{k \in
  K}$, where the {\em domain} of $\sigma$ is $\{ x_k\}_{k \in
  K}$, $K$ being a finite set of indices.  Application of a
substitution $\sigma$ to a term $t$ is written $\sigma(t)$.

Given a conjunction of $\Sigma$-equalities $\Phi$ of the form $\bigwedge_{k \in K}
s_k = t_k$ ($K$ being a finite set of indices), a {\em unifier} of $\Phi$ is a substitution $\mu$ such that
$\mu(s_k) = \mu(t_k)$ for each $k \in K$. A conjunction of
$\Sigma$-equalities $\Gamma$ of the form $\bigwedge_{k \in K} x_k=t_k$ is said to
be a {\em solved form} if for each $k \in K$, $x_k$ is a variable occurring
only once in $\Gamma$. Given any conjunction of
$\Sigma$-equalities $\Phi$, a syntactic unification algorithm computes $\bot$
if $\Phi$ has no unifier, otherwise it computes a solved form $\Gamma$
such that $\Phi$ and $\Gamma$ have the same set of unifiers, equivalently $T(\Sigma,V) \models \Phi \Leftrightarrow \Gamma$.
If $\Gamma$ denotes the solved form $\bigwedge_{k \in K} x_k = t_k$, then the corresponding
substitution $\mu = \{ x_k \mapsto t_k \}_{k \in K}$ is a unifier of
$\Gamma$. Actually, $\mu$ is a most general unifier of
$\Gamma$, meaning that all the unifiers of $\Gamma$ are instances of
$\mu$.

We begin with the following lemma, which is based on \cite{CFR19}:

\begin{lemma}
  \label{lem:trees-Shostak}
  Assume $\Phi$ is any conjunction of $\consig{\Sigma}$-equalities and $\Delta$ is any conjunction of $\consig{\Sigma}$-disequalities. 

  \begin{itemize}
  \item $\Phi \wedge \Delta$ is $\TREE_{\Sigma}$-satisfiable iff there exists a most general unifier $\mu = \{ x_k \mapsto t_k \}_{k \in K}$ of $\Phi$ such that
    for any $v \neq w$ in $\Delta$, $\mu(v) \neq \mu(w)$.
\item If $\mu = \{ x_k \mapsto t_k \}_{k \in K}$ is a most general unifier of $\Phi$, then the conjunction of $\consig{\Sigma}$-equalities $\Gamma = (\bigwedge_{k \in K} x_k = t_k)$ is such that
    \begin{itemize}
    \item for each $k \in K$, $x_k$ occurs only once in $\Gamma$,
    \item $\Phi$ and $\Gamma$ are $\TREE_{\Sigma}$-equivalent. 
    \end{itemize}
  \end{itemize}
\end{lemma}  

Let us introduce the notion of $is$-constraint that will be used in the proof of \Cref{groundcomplete}.
Given a finite set of $\structsorts$-sorted variables $V$, an
$is$-constraint over $V$ is set of literals $\rho \suq \{ is_{c}(x) ~|~ x \in V,c\in\constructors
\}$ such that for any $is_c(x) \in \rho$,
$c:\sigma_{1}\times\ldots\times\sigma_{n}\ra\sigma \in \constructors$ if $x$ is
of sort $\sigma$, and for every $x\in V$, $is_c(x)\in\rho$ for some $c$.
The set of $is$-constraints over $V$ is denoted by $IS(V)$.
Given an $is$-constraint $\rho$, $\rho_{eq} = \{ x = c(y_1,\dots,y_n) ~|~
is_c(x) \in \rho \}$ such that all the variables $y_1,\dots,y_n$ in $\rho_{eq}$
are fresh.

Now we proceed with the proof of \Cref{groundcomplete}.

Assume $\varphi$ is any conjunction of flat
$\Sigma$-literals including an arrangement over the variables in $\varphi$.
Consider the set of variables $GV(\varphi)$ defined as
$$\{ x ~|~ is_c(x) \in \varphi \}
\cup
\{ x ~|~ \neg is_c(x) \in \varphi \}
\cup
\{ y ~|~ x=s_{c,i}(y) \in \varphi, s_{c,i} \in \selectors \}$$
excluding all the variables in 
$$\{ y ~|~ x=s_{c,i}(y), y = d(\dots) \in \varphi, s_{c,i} \in \selectors, d \in \constructors, d \neq c \}.$$
We want to build a formula equivalent to $\varphi$ but including at least one $\sigma$-sorted variable for each $\sigma \in \elemsorts$. For this reason, 
let us denote $\varphi_{te}$ a conjunction of trivial equalities $x_\sigma = x_\sigma$, one for every $\sigma \in \elemsorts$ such that $\fv{\sigma}{\varphi} =  \emptyset$, $x_\sigma$ being a fresh $\sigma$-sorted variable.
If $GV(\varphi)= \emptyset$, define $\varphi_1 = \varphi \wedge \varphi_{te}$. Otherwise, define $\varphi_1$ as follows:
$$\varphi_1 = \bigvee_{\rho \in IS(GV(\varphi))} w(\varphi,\rho_{eq}) \wedge \rho_{eq} \wedge \varphi_{te}$$
where $w(\varphi,\rho_{eq})$ is inductively defined as follows:

\begin{enumerate}

\item $w(\emptyset,\rho_{eq}) = \emptyset$,

\item $w(\{ is_c(x) \} \cup \varphi,\rho_{eq}) = w(\varphi,\rho_{eq})$ if $x=c(y_1,\dots,y_n) \in \rho_{eq}$,

\item $w(\{ \neg is_c(x) \} \cup \varphi,\rho_{eq}) = \bot$ if $x=c(y_1,\dots,y_n) \in \rho_{eq}$,

\item $w(\{ is_d(x) \} \cup \varphi,\rho_{eq}) = \bot$ if $x=c(y_1,\dots,y_n) \in \rho_{eq}$ and $c \neq d$,

\item $w(\{ \neg is_d(x) \} \cup \varphi,\rho_{eq}) = w(\varphi,\rho_{eq})$ if $x=c(y_1,\dots,y_n) \in \rho_{eq}$ and $c \neq d$,

\item $w(\{ y = s_{c,i}(x) \} \cup \varphi,\rho_{eq}) = \{ y = y_i \} \cup w(\varphi,\rho_{eq})$ if $x=c(y_1,\dots,y_n) \in \rho_{eq}$,
  

\item otherwise, $w(l \cup \varphi,\rho_{eq}) = \{ l \} \cup w(\varphi,\rho_{eq})$.

\end{enumerate}
where all the $y_{i}$ are fresh.
Note that the above construction is similar to the one given in~\cite{CFR19} (see Proposition~4 in~\cite{CFR19}). One can observe that $\varphi \wedge \rho_{eq}$ and $w(\varphi,\rho_{eq}) \wedge \rho_{eq}$ are $\TREE_{\Sigma}^{\ast}$-equivalent. In particular, for the case $(6.)$ above, it follows from the projection axiom $\mathit{Proj}$ in $\TREE_{\Sigma}^{\ast}$. In addition the guessing of $is$-constraint preserves the $\TREE_{\Sigma}^{\ast}$-equivalence since $\TREE_{\Sigma}^{\ast}$ includes the extensionality axioms ${\mathit{Ext}_{1}}$ and ${\mathit{Ext}_{2}}$. Thus $\varphi$ and $\exists \ora{w}  . \varphi_1$ are $\TREE_{\Sigma}^{\ast}$-equivalent for $\ora{w} = \fv{}{\varphi_1} \backslash \fv{}{\varphi}$. 

\medskip

For any set of literals $\phi$, let us define
$$Min(\phi) = \fv{}{\phi} \backslash \{ x ~|~ x=c(\dots) \mbox{ occurs in } \phi \}.$$
Starting from $\varphi_1$, consider the following sequences of formulas, obtained by guessing $is$-constraints for ``minimal'' variables of finite sorts: 
$$\varphi_{j+1} = \bigvee_{\rho \in IS(\bigcup_{\sigma \in \Fin(\Sigma)} Min_\sigma(\varphi_j))} \varphi_j \wedge \rho_{eq}$$
By definition of $\Fin(\Sigma)$, there exists necessarily some $j'$ such that
the set of variables $\bigcup_{\sigma \in \Fin(\Sigma)} Min_\sigma(\varphi_{j'})$ is empty. In that case, let us define $\varphi' = \varphi_{j'}$.

\bigskip
It is routine to show that
$\varphi$ and $\exists \ora{w}  . \varphi'$ are $\Ta_{\Sigma}$-equivalent for  the set of fresh variables $\ora{w} = \fv{}{\varphi'} \backslash \fv{}{\varphi}$, using the following facts:
\begin{itemize}
  
\item $\valid{\Ta_{\Sigma}} \TREE_{\Sigma}^{\ast}$,

\item as shown above, $\varphi$ and $\exists \ora{w}  . \varphi_1$ are $\TREE_{\Sigma}^{\ast}$-equivalent for the set of fresh variables $\ora{w} = \fv{}{\varphi_1} \backslash \fv{}{\varphi}$,

\item $\varphi_j$ and $\exists\ora{w}  . \varphi_{j+1}$ are $\TREE_{\Sigma}^{\ast}$-equivalent, for $\ora{w} = \fv{}{\varphi_{j+1}} \backslash \fv{}{\varphi_j}$ and any $j=1,\dots, j'-1$, since $\TREE_{\Sigma}^{\ast}$ includes the extensionality axioms ${\mathit{Ext}_{1}}$ and ${\mathit{Ext}_{2}}$. 
  
\end{itemize}  
  
Let us now show that $\varphi'$ is $\Ta_{\Sigma}$-satisfiable iff $\varphi'$ is $\TREE_{\Sigma}$-satisfiable.

$(\Rightarrow)$ directly follows from \Cref{sound}.

$(\Leftarrow)$ If $\varphi'$ is $\TREE_{\Sigma}$-satisfiable, there exists a $\TREE_{\Sigma}$-interpretation $\Aa$ and a disjunct $\psi$ of $\varphi'$ such that $\Aa \models \psi$. By construction of $\varphi'$, $\psi$ is a conjunction
$\psi_{\constructors} \wedge \psi_{\selectors}$ where
\begin{itemize}
\item $\psi_{\constructors}$ is a conjunction of $\consig{\Sigma}$-literals,
\item $\psi_{\selectors}$ is a conjunction of equalities of the form $x = s_{c,i}(y)$.  
\end{itemize}  

Since $\psi$ holds in a $\TREE_{\Sigma}$-interpretation, the conjunction of $\consig{\Sigma}$-equalities in $\psi_{\constructors}$ has a most general unifier. By Lemma~\ref{lem:trees-Shostak}, $\psi_{\constructors}$ is $\TREE_{\Sigma}$-equivalent to a conjunction of literals $\Gamma \wedge \Delta$ such that

\begin{itemize}
  
\item $\Gamma$ is a conjunction of equalities $\bigwedge_{k \in K} x_k = t_k$ such that for each $k \in K$, $x_k$ is a variable occurring only once in $\Gamma$,
\item $\Delta$ is the conjunction of disequalities in $\psi$,

\item given the substitution $\mu = \{ x_k \mapsto t_k \}_{k \in K}$, for any $v \neq w$ in $\Delta$, $\mu(v) \neq \mu(w)$.
  
\end{itemize}

Consider the set of variables $MV = \{ x \in \fv{\structsorts}{\varphi'} ~|~ \mu(x) = x \}$. Since the sorts of variables in $MV$ are all inductive, there exists a substitution $\alpha$ from $MV$ to $T(\consig{\Sigma}, \fv{\elemsorts}{\varphi'})$ such that for any $x, y \in MV$, 
$ \alpha(x)^\Aa  = \alpha(y)^\Aa$ iff $x=y$. According to this substitution $\alpha$, we have for any terms $t,u \in T(\consig{\Sigma},MV \cup \fv{\elemsorts}{\varphi'})$, $\alpha(t)^\Aa = \alpha(u)^\Aa$ iff $t = u$. In particular, we have for any $k,k' \in K$,
$$\alpha(\mu(x_k))^\Aa = \alpha(\mu(x_{k'}))^\Aa \mbox{ iff } \mu(x_k) = \mu(x_{k'}).$$
It is always possible to choose $\alpha$ such that
for any $x,y \in MV$, $x \neq y$, we have
$$|\depth(\alpha(x)) - \depth(\alpha(y)) | > \max \{ \depth(t_k) \}_{k \in K}.$$ 
According to the assumption on $\alpha$, it is impossible to have 
$\alpha(\mu(x_k))^\Aa = \alpha(\mu(x))^\Aa$ for some $k \in K$ and some $x \in MV$. Consequently,
we have for any $x,y \in \fv{\structsorts}{\varphi'}$, 
$$\alpha(\mu(x))^\Aa = \alpha(\mu(y))^\Aa \mbox{ iff } \mu(x) = \mu(y).$$

Let us now consider $\Ba \in \Ta_{\Sigma}$ such that
\begin{itemize}

\item for any $\sigma \in \elemsorts$, $\sigma^{\Ba} = \{ e^\Aa ~|~ e \in \fv{\sigma}{\varphi'} \}$,

\item for any $x \in \fv{\structsorts}{\varphi'}$, $x^\Ba = \alpha(\mu(x))^\Aa$,

\item for any $e \in \fv{\elemsorts}{\varphi'}$, $e^\Ba = e^\Aa$.
  
\end{itemize}
One can observe that $\Ba \models \Gamma \wedge \Delta$ since
\begin{itemize}
\item for any $x_k = t_k$ in $\Gamma$, $\mu(x_k) = \mu(t_k)$ and so $x_k^\Ba = t_k^\Ba$,
\item for any $v \neq w$ in $\Delta$, $\mu(v) \neq \mu(w)$ and so $v^\Ba \neq w^\Ba$.
\end{itemize}

Since $\Gamma \wedge \Delta$ and $\psi_{\constructors}$ are $\TREE_{\Sigma}$-equivalent and 
$ \valid{\Ta_{\Sigma}} \TREE_{\Sigma}$, 
we have $\Ba \models \psi_{\constructors}$.

Let us now consider the conjunction $\psi_{\selectors}$ that contains only equalities of the form
$x = s_{c,i}(y)$. By construction of $\varphi'$, the term $\mu(y)$ is necessarily rooted by a constructor $d \in \constructors$, $d \neq c$. Thus $s_{c,i}^\Ba$ can be defined arbitrarily on $y^\Ba$ since $y^\Ba$ is a standard tree rooted by some constructor $d$ different from $c$. In particular, we can define $s_{c,i}^\Ba$ such that $s_{c,i}^\Ba(y^\Ba) = x^\Ba$. Using this interpretation $\Ba$ for the selectors, we have $\Ba \models \psi_{\selectors}$.

Since $\Ba \models \psi_{\constructors}$ and $\Ba \models \psi_{\selectors}$, we get $\Ba \models \psi$. Since $\psi$ is some disjunct of $\varphi'$, we can conclude that $\Ba \models \varphi'$. 

\qed

\end{report}

\end{document}